\documentclass[runningheads,a4paper]{llncs}

\usepackage{amssymb}
\usepackage{graphicx}

\usepackage{url}
\newcommand{\keywords}[1]{\par\addvspace\baselineskip
\noindent\keywordname\enspace\ignorespaces#1}

\newtheorem{fact}[lemma]{Lemma}
\newcommand{\Fact}{Lemma }
\newcommand{\Facts}{Lemmas\ }

\newcommand{\eqref}[1]{(\ref{#1})}

\newcommand{\conf}[1]{\ensuremath{C(#1)}}

\newcommand{\EXP}[1]{\ensuremath{{\mathbf{E}\left[#1\right]}}}
\newcommand{\cE} {\ensuremath{\mathcal E}}

\newcommand{\Csc}[1]{\ensuremath{Cost(#1)}}

\newcommand{\xlambda}{\ensuremath{\lambda}}

\newcommand{\xstar}{\ensuremath{x^*}}
\newcommand{\ystar}{\ensuremath{y^*}}

\newcommand{\Deltaf}[1]{\ensuremath{\Delta(#1)}}

\newcommand{\fgxnatwidth}{274}
\newcommand{\fgxnatheight}{318}

\newcommand{\plotnatwidth}{555}
\newcommand{\plotnatheight}{555}

\newcommand{\elfarol}{El Farol game}
\newcommand{\elfarolws}{El Farol game\ }

\newcommand{\extelfarol}{$(c,s_1,s_2)$-El Farol game}
\newcommand{\extelfarolws}{$(c,s_1,s_2)$-El Farol game\ }

\newcommand{\mv}{\emph{Mediation Value}}
\newcommand{\mvsym}{\emph{MV}}

\newcommand{\ev}{\emph{Enforcement Value}}
\newcommand{\evsym}{\emph{EV}}

\def\D{{\cal D}}

\newcommand{\SpacedNe}{Nash equilibrium\ }
\newcommand{\Ne}{Nash equilibrium}

\newcommand{\paperTitle}{The Power of Mediation in an Extended El Farol Game}

\newcommand{\paperTitleAbbr}{The Power of Mediation in an Extended El Farol Game}

\newcommand{\authorsAbbr}{D. Mitsche, G. Saad and J. Saia}

\date{}

\begin{document}
\mainmatter  

\mainmatter              

\author{Dieter Mitsche\inst{1} \and George Saad\inst{2} \and Jared Saia\inst{2}}

\institute{Laboratoire Dieudonn\'{e}, UMR CNRS 7351,Universit\'{e} de Nice,\\
\email{dmitsche@unice.fr}
\and
Department of Computer Science,
University of New Mexico, 
\\
\email{\{george.saad,saia\}@cs.unm.edu}
}

\title{\paperTitle}
\titlerunning{\paperTitleAbbr}

%
%

%
\authorrunning{\authorsAbbr}
\tocauthor{Dieter Mitsche, George Saad, and  Jared Saia}



%
%

\toctitle{\paperTitleAbbr}
\maketitle

\begin{abstract}
A mediator implements a correlated equilibrium when it proposes a strategy to each player confidentially such that the mediator's proposal is the best interest for every player to follow.
In this paper, we present a mediator that implements the best correlated equilibrium for an extended \elfarolws with symmetric players. 
The extended El Farol game we consider incorporates both negative and positive network effects.

We study the degree to which this type of mediator can decrease the overall social cost. 
In particular, we give an exact characterization of \mv ~(\mvsym) and \ev ~(\evsym) for this game. 
\mvsym ~is the ratio of the minimum social cost over all Nash equilibria to the minimum social cost over all mediators of this type, and \evsym ~is the ratio of the minimum social cost over all mediators of this type to the optimal social cost.
This sort of exact characterization is uncommon for games with both kinds of network effects.  
An interesting outcome of our results is that  both the $\mvsym$ and $\evsym$ values can be unbounded for our game.

\end{abstract}

\section{Introduction}
When players act selfishly to minimize their own costs, the outcome with respect to the total social cost may be poor.
The Price of Anarchy \cite{Koutsoupias1999} measures the impact of selfishness on the social cost and is defined as the ratio of the worst social cost over all Nash equilibria to the optimal social cost.
In a game, with a high Price of Anarchy, one way to reduce social cost is to find a mediator of expected social cost less than the social cost of any \Ne.

In the literature, there are several types of mediators \cite{Ashlagi:2007,Diaz2009,Forgo2010,Forgo2010-2,Monderer:2009,Peleg:2007,RT,RT3,RT2,Tennenholtz:2008}. 
In this paper, we consider only the type of mediator that implements a correlated equilibrium (CE) \cite{Aumann}.

A mediator is a trusted external party that suggests a strategy to every player separately and privately so that each player has no gain to choose another strategy assuming that the other players conform to the mediator's suggestion.

The algorithm that the mediator uses is known to all players. 
However, the mediator's random bits are unknown.
We assume that the players are symmetric in the sense that they have the same utility function and the probability the mediator suggests a strategy to some player is independent of the identity of that player.

Ashlagi et al. \cite{AshlagiMT08} define two metrics to measure the quality of a mediator: the mediation value (\mvsym) and the enforcement value (\evsym).
In our paper, we compute these values, adapted for games where players seek to minimize the social cost. 
The \mv ~is defined as the ratio of the minimum social cost over all Nash equilibria to the minimum social cost over all mediators.
The \ev ~is the ratio of the minimum social cost over all mediators to the optimal social cost.

A mediator is optimal when its expected social cost is minimum over all mediators.
Thus, the \mv ~measures the quality of the optimal mediator with respect to the best \Ne ; and the \ev ~measures the quality of the optimal mediator with respect to the optimal social cost.

\subsection{El Farol Game}\label{sec:elfarol}
First we describe the traditional \elfarolws \cite{Arthur,CPG,CMO,LAKUA}.
El Farol is a tapas bar in Santa Fe.  
Every Friday night, a population of people decide whether or not to go to the bar. 
If too many people go, they will all have a worse time than if they stayed home, since the bar will be too crowded.
That is a negative network effect \cite{David:2010}.

Now we provide an extension of the traditional \elfarol, where both negative and positive network effects \cite{David:2010} are considered. The positive network effect is that if too few people go, those that go will also have a worse time than if they stayed home.

\subsubsection{Motivation.}

Our motivation for studying this problem comes from the following discussion in \cite{David:2010}.

\medskip
\noindent
\emph{``It's important to keep in mind, of course, that many real situations in fact display both kinds of [positive and negative] externalities - some level of participation by others is good, but too much is bad.  For example, the El Farol Bar might be most enjoyable if a reasonable crowd shows up, provided it does not exceed 60.  Similarly, an on-line social media site with limited infrastructure might be most enjoyable if it has a reasonably large audience, but not so large that connecting to the Web site becomes very slow due to the congestion."}

\medskip
We note that our El Farol extension is one of the simplest,  non-trivial problems for which a mediator can improve the social cost.  Thus, it is useful for studying the power of a mediation.


\subsubsection{Formal Definition of the Extended El Farol Game.}

\begin{figure}
\centerline{
\includegraphics[width=0.4\textwidth,natwidth=\fgxnatwidth,natheight=\fgxnatheight]{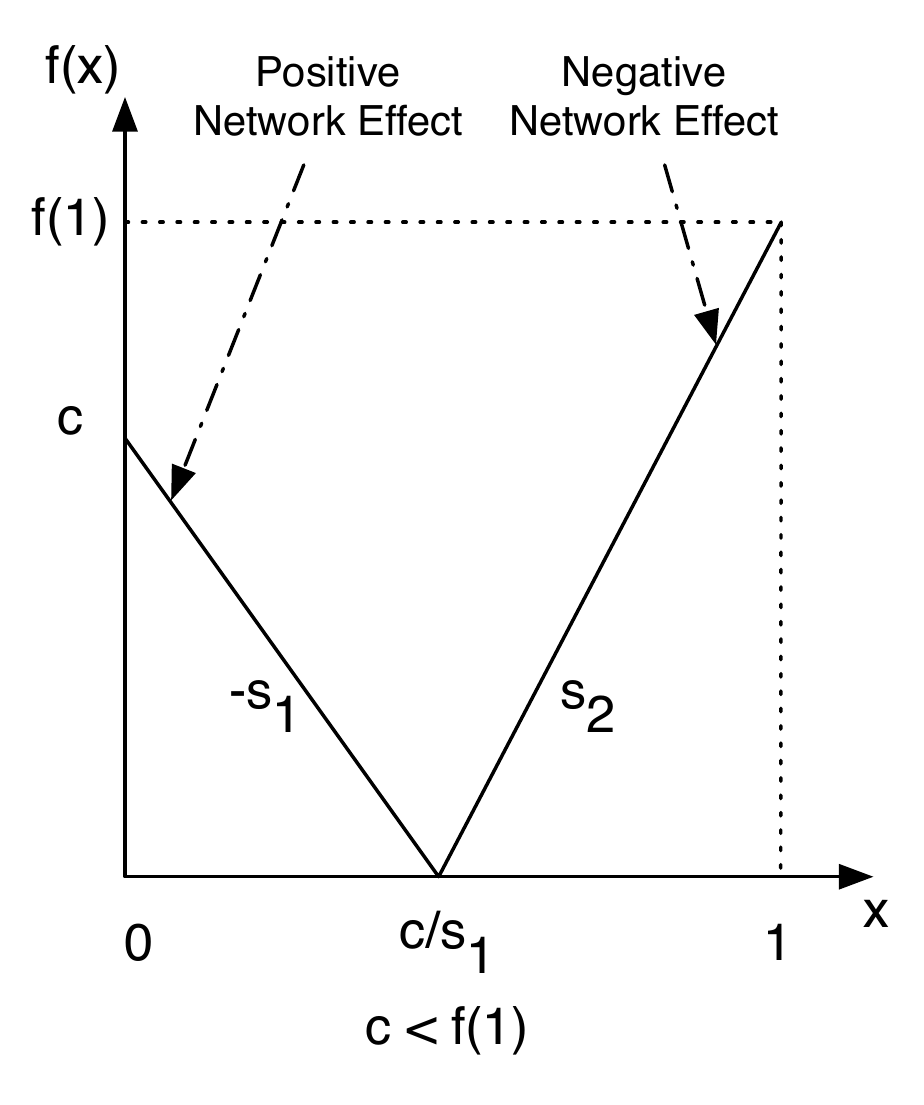}
\includegraphics[width=0.4\textwidth,natwidth=\fgxnatwidth,natheight=\fgxnatheight]{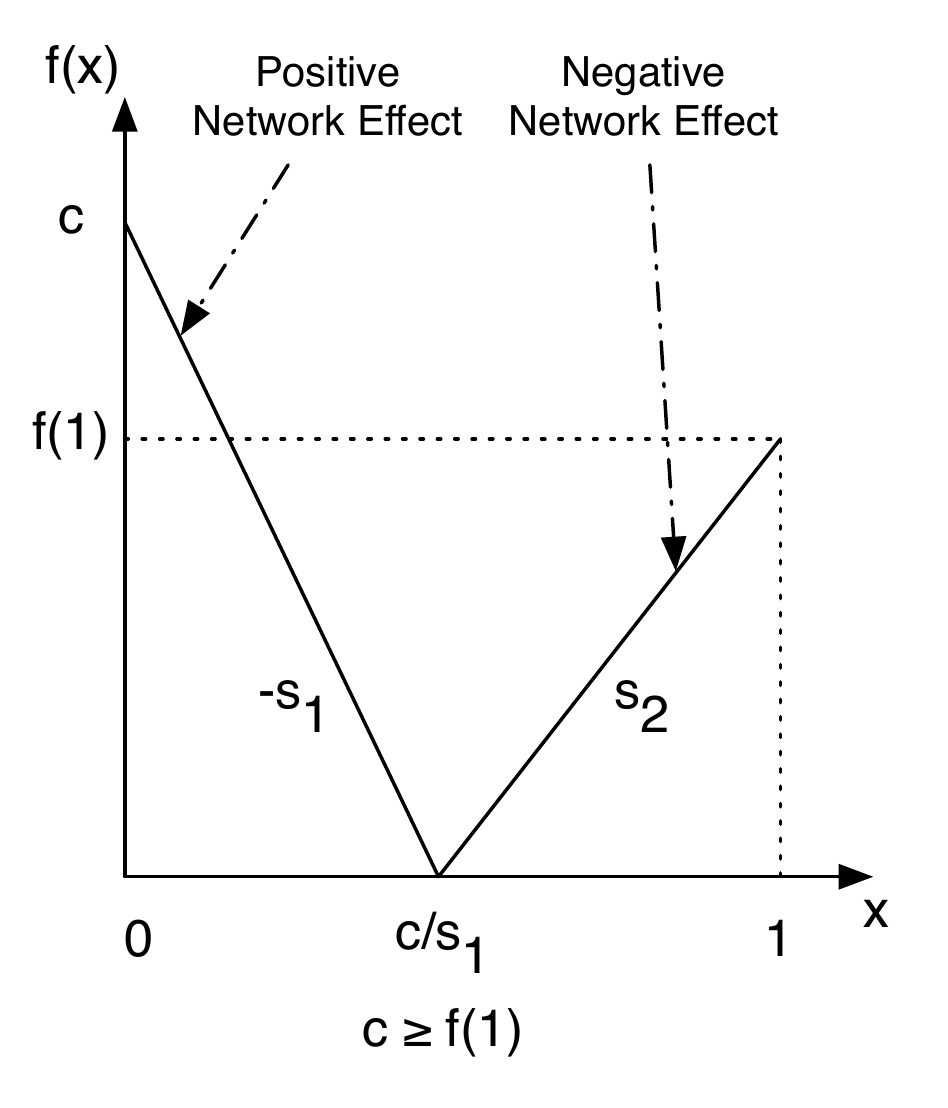}
}
\caption{The individual cost to go $f(x)$.}
\label{fig:fgx}
\end{figure}

We now formally define our game, which is non-atomic \cite{aumann1974values,Schmeidler1973}, in the sense that no individual player has significant influence on the outcome; moreover, the number of players is very large tending to infinity. 
The \extelfarolws has three parameters $c, s_1$ and $s_2$, where $0 < c < s_1$ and $s_2 > 0$.
If $x$ is the fraction of players to go, then the cost $f(x)$ for any player to go is as follows:
\begin{equation}\label{eq:fgx}
f(x) = \left\{ 
\begin{array}{l l}
  c- s_1 x & \quad \mbox{$0 \leq x \leq \frac{c}{s_1}$,}\\
  s_2 (x - \frac{c}{s_1}) & \quad \mbox{$\frac{c}{s_1} \leq x \leq 1$.}\\ 
\end{array} \right.
\end{equation}
and the cost to stay is 1.
The function $f(x)$ is illustrated in the two plots of Figure \ref{fig:fgx}.

\subsubsection{Our Contributions.}
The main contributions of our paper are threefold:
\begin{itemize}
\item We design an optimal mediator, which implements the best correlated equilibrium for an extension of the \elfarolws with symmetric players. 
Notably, this extension incorporates both negative and positive network effects.

\item We give an exact characterization of the \mv ~(\mvsym) and the \ev ~(\evsym) for our game. 
\item We show that both the $\mvsym$ and $\evsym$ values can be unbounded for our game.
\end{itemize}

\subsubsection{Paper Organization.}
In Section \ref{relatedwork}, we discuss the related work.
Section \ref{sec:definitions} states the definitions and notations that we use in the \elfarol.
Our results are given in Section \ref{sec:our_results}, where we show our main theorem that characterizes the best correlated equilibrium, and we compute accordingly the \mv ~and the \ev .
Finally, Section~\ref{sec:conclusion} concludes the paper and discusses some open problems.

\section{Related Work}\label{relatedwork}


\subsection{Mediation Metrics}

Christodoulou and Koutsoupias \cite{Christodoulou:2005} analyze the price of anarchy and the price of stability for Nash and correlated equilibria in linear congestion games. 
A consequence of their results is that the \evsym ~for these games is at least $1.577$ and at most $1.6$, and the \mvsym ~is at most $1.015$.

Brandt et al. \cite{Brandt:2007} compute the mediation value and the enforcement value in ranking games. In a ranking game, every outcome is a ranking of the players, and each player strictly prefers high ranks over lower ones \cite{Brandt:2006}. 
They show that for the ranking games with $n>2$ players, $\evsym = n-1$. They also show that $\mvsym = n-1$ for $n>3$ players, and for $n=3$ players where at least one player has more than two actions.


The authors of \cite{Diaz2009} design a mediator that implements a correlated equilibrium for a virus inoculation game \cite{ACY,MSW}. 
In this game, there are $n$ players, each corresponding to a node in a square grid. Every player has either to inoculate itself (at a cost of $1$) or to do nothing and risk infection, which costs $L>1$. 
After each node decides to inoculate or not, one node in the grid selected uniformly at random is infected with a virus.
Any node, $v$, that chooses not to inoculate becomes infected if there is a path from the randomly selected node to $v$ that traverses only uninoculated nodes.
A consequence of their result is that \evsym ~is $\Theta(1)$ and \mvsym ~is $\Theta((n/L)^{1/3})$ for this game. 

Jiang et al. \cite{Xin13a} analyze the price of miscoordination (PoM) and the price of sequential commitment (PoSC) in security games, which are defined to be a certain subclass of Stackelberg games. A consequence of their results is that \mvsym ~is unbounded in general security games and it is at least $4/3$ and at most $\frac{e}{e-1} \thickapprox 1.582$ in a certain subclass of security games.


We note that a poorly designed mediator can make the social cost worse than what is obtained from the Nash equilibria.  
Bradonjic et al. \cite{Bradonjic:2009} describe the \emph{Price of Mediation} ($PoM$) which is the ratio of the social cost of the worst correlated equilibrium to the social cost of the worst \Ne.  
They show that for a simple game with two players and two possible strategies, $PoM$ can be as large as $2$.  
Also, they show for games with more players or more strategies per player that $PoM$ can be unbounded.

\subsection{Finding and Simulating a Mediator}
Papadimitriou and Roughgarden \cite{Papadimitriou:2008} develop polynomial time algorithms for finding correlated equilibria in a broad class of succinctly representable multiplayer games.
Unfortunately, their results do not extend to non-atomic games; moreover, they do not allow for direct computation of \mvsym ~and \evsym, even when they can find the best correlated equilibrium. 




Abraham et al.~\cite{ADGH,ADH} describe a distributed algorithm that enables a group of players to simulate a mediator. 
This algorithm works robustly with up to linear size coalitions, and up to a constant fraction of adversarial players. 
The result suggests that the concept of mediation can be useful even in the absence of a trusted external party. 

\subsection{Other Types of Mediators}


In all equilibria above, the mediator does not act on behalf of the players. However, a more powerful type of mediators is described in \cite{Ashlagi:2007,Forgo2010,Forgo2010-2,Monderer:2009,Peleg:2007,RT,RT3,RT2,Tennenholtz:2008}, where a mediator can act on behalf of the players that give that right to it. 

For multistage games, the notion of the correlated equilibrium is generalized to the communication equilibrium in \cite{Forges1986,Myerson1986}. In a communication equilibrium, the mediator implements a multistage correlated equilibrium; in addition, it communicates with the players privately to receive their reports at every stage and selects the recommended strategy to each player accordingly.

\section{Definitions and Notations}\label{sec:definitions}
%


Now we state the definitions and notations that we use in the \elfarol.

\begin{definition}
\emph{A configuration} $\conf{x}$ characterizes that a fraction of players, $x$, is being advised to go; and the remaining fraction of players, $(1-x)$, is being advised to stay. 
\end{definition}

\begin{definition}
\emph{A configuration distribution $\D\{(\conf{x_1},p_1),..,(\conf{x_k},p_k)\}$} is a probability distribution over $k\geq 2$ configurations, where $(\conf{x_i},p_i)$ represents that configuration $\conf{x_i}$ is selected with probability $p_i$, for $1\leq i \leq k$.
Note that $0\leq x_i\leq 1$, $0<p_i<1$, $\sum^k_{i=1} p_i = 1$ and if $x_i=x_j$ then $i=j$ for $1\leq i,j\leq k$.
\end{definition}

For any player $i$, let $\cE^i_G$ be the event that player $i$ is advised to go, and $C^i_G$ be the cost for player $i$ to go (when all other players conform to the advice). 
Also let $\cE^i_S$ be the event that player $i$ is advised to stay, and $C^i_S$ be the cost for player $i$ to stay. 
Since the players are symmetric, we will omit the index $i$.

A configuration distribution, $\D\{(\conf{x_1},p_1),..,(\conf{x_k},p_k)\}$, is a correlated equilibrium iff 
\begin{eqnarray*}
\EXP{C_S |\cE_G} \geq \EXP{C_G |\cE_G}, \\
\EXP{C_G |\cE_S} \geq \EXP{C_S |\cE_S}.
\end{eqnarray*}

\begin{definition}
\emph{A mediator} is a trusted external party that uses a configuration distribution to advise the players such that this configuration distribution is a correlated equilibrium.
The set of configurations and the probability distribution are known to all players. 
The mediator selects a configuration according to the probability distribution.
The advice the mediator sends to a particular player, based on the selected configuration, is known only to that player.
\end{definition}

Throughout the paper, we let $n$ be the number of players.

\section{Our Results}\label{sec:our_results}

In our results, we assume that \emph{the cost to stay} is $1$; we justify this assumption at the end of this section.
Our first results in Lemmas \ref{lem:socialoptimum} and \ref{lem:bestnash} are descriptions of the optimal social cost and the minimum social cost over all Nash equilibria for our extended \elfarol. 
We next state our main theorem which characterizes the best correlated equilibrium and determines the \mv ~and \ev.

\begin{lemma}\label{lem:socialoptimum}
For any \extelfarol , the optimal social cost is $(\ystar f(\ystar)+(1-\ystar))n$, 
where \[\ystar = \left\{ 
\begin{array}{l l}
  \frac{1}{2}(\frac{c}{s_1}+\frac{1}{s_2}) & \quad \mbox{if $\frac{c}{s_1}  \leq  \frac{1}{2}(\frac{c}{s_1}+\frac{1}{s_2}) \leq 1$,}\\
  \frac{c}{s_1} & \quad \mbox{if $ \frac{1}{s_2} < \frac{c}{s_1}$,}\\ 
  1 & \quad \mbox{$otherwise$.}\\ 
\end{array} \right. \] 
\end{lemma}

\begin{proof}
By Equation \eqref{eq:fgx}, $f(x)$ has two cases. 
Let $f_1(x)$ be $f(x)$ for $x \in [0,\frac{c}{s_1}]$, and let $f_2(x)$ be $f(x)$ for $x \in [\frac{c}{s_1}, 1]$.
Also let $h_1(x)$ be the social cost when $0 \leq x \leq \frac{c}{s_1}$, and let $h_2(x)$ be the social cost when $\frac{c}{s_1} \leq x \leq 1$. 
Thus, $h_1(x) = (xf_1(x)+(1-x))n$ and $h_2(x) = (xf_2(x)+(1-x))n$.

We know that $h_1(x)$ is minimized at $x=\frac{c}{s_1}$.
In addition, we know that $h_2(x)$ is a quadratic function with respect to $x$, and thus it has one minimum over $x \in [\frac{c}{s_1}, 1]$ at $x=\ystar$, where:
 \[\ystar = \left\{ 
\begin{array}{l l}
  \frac{1}{2}(\frac{c}{s_1}+\frac{1}{s_2}) & \quad \mbox{if $\frac{c}{s_1}  \leq  \frac{1}{2}(\frac{c}{s_1}+\frac{1}{s_2}) \leq 1$,}\\
  \frac{c}{s_1} & \quad \mbox{if $ \frac{1}{2}(\frac{c}{s_1}+\frac{1}{s_2}) < \frac{c}{s_1}$,}\\ 
  1 & \quad \mbox{$ otherwise$.}\\ 
\end{array} \right. \] 

Let $h^*$ be the optimal social cost. Then $h^* = min (h_1(\frac{c}{s_1}), h_2(\ystar))$. Since $f_1(\frac{c}{s_1}) = f_2(\frac{c}{s_1})$, we have $h_1(\frac{c}{s_1}) = h_2(\frac{c}{s_1})$. Hence, $h^* = min (h_2(\frac{c}{s_1}), h_2(\ystar))$. This implies that $h^* = h_2(\ystar)$. \qed
\end{proof}

\begin{lemma}\label{lem:bestnash}
For any \extelfarol, if $f(1)\geq 1$, then the best Nash equilibrium is at which the cost to go in expectation is equal to the cost to stay; otherwise, the best Nash equilibrium is at which all players would rather go. 
The social cost of the best \SpacedNe is $\min(n, f(1) \cdot n)$.
\end{lemma}
\begin{proof}
There are two cases for $f(1)$ to determine the best \Ne.\\
\textbf{Case 1:} $f(1)\geq1$.
Let $N_y$ be a Nash equilibrium with the minimum social cost over all Nash equilibria and with a $y$-fraction of players that go in expectation.
If $f(y)>1$, then at least one player of the $y$-fraction of players would rather stay.
Also if $f(y)<1$, then at least one player of the $(1-y)$-fraction of players would rather go.
Thus, we must have $f(y) = 1$.
Assume that each player has a mixed strategy, where player $i$ goes with probability $y_i$. 
Recall that $N_y$ has a $y$-fraction of players that go in expectation. Thus, $y=\frac{1}{n}\sum^{n}_{i=1} y_i$.
Then the social cost is $\sum^{n}_{i=1}(y_if(y)+(1-y_i))$, or equivalently, $n$.\\
\textbf{Case 2:} $f(1)<1$.
In this case, the best \SpacedNe is at which all players would rather go, with a social cost of $f(1) \cdot n$.

Therefore, the social cost of the best \SpacedNe is $min(n, f(1) \cdot n)$. \qed
\end{proof}

\begin{theorem}\label{thm:optimal}
For any \extelfarolws, if $c \leq 1$, then the best correlated equilibrium is the best \Ne; otherwise, the best correlated equilibrium is $\D\{(\conf{0},p),(\conf{\xstar},1-p)\}$, where 
$
\xlambda(c,s_1,s_2) = c(\frac{1}{s_1} + \frac{1}{s_2}) - \sqrt{\frac{c(\frac{1}{s_1} + \frac{1}{s_2})(c-1)}{s_2}},
$
$$
\xstar = \left\{ 
\begin{array}{l l}
  \xlambda(c,s_1,s_2) & \quad \mbox{if $\frac{c}{s_1}  \leq  \xlambda(c,s_1,s_2) < 1$,}\\
  \frac{c}{s_1} & \quad \mbox{if $ \xlambda(c,s_1,s_2) < \frac{c}{s_1}$,}\\ 
  1 & \quad \mbox{$ otherwise$.}\\ 
\end{array} \right.
$$
and
$p = \frac{(1-\xstar)(1-f(\xstar))}{(1-\xstar)(1-f(\xstar))+c-1}$. 
Moreover,

\begin{enumerate}

\item[1)] the expected social cost is
$
(p+(1-p)(\xstar f(\xstar)+(1-\xstar)))n
$,
\item[2)] the Mediation Value $(\mvsym)$ is
$
\frac{\min(f(1), 1)}{p+(1-p)(\xstar f(\xstar)+(1-\xstar))}
$ and

\item[3)] the Enforcement Value $(\evsym)$ is
$
\frac{p+(1-p)(\xstar f(\xstar)+(1-\xstar))}{\ystar f(\ystar)+(1-\ystar)},
$
where 
$$
\ystar = \left\{ 
\begin{array}{l l}
  \frac{1}{2}(\frac{c}{s_1}+\frac{1}{s_2}) & \quad \mbox{if $\frac{c}{s_1}  \leq  \frac{1}{2}(\frac{c}{s_1}+\frac{1}{s_2}) \leq 1$,}\\
  \frac{c}{s_1} & \quad \mbox{if $ \frac{1}{s_2} < \frac{c}{s_1}$,}\\ 
  1 & \quad \mbox{$ otherwise$.}\\ 
\end{array} \right. .
$$
\end{enumerate}

\end{theorem}

Due to the space constraints, the proof of this theorem is not given here.

The following corollary shows that for $c > 1$, if $\xlambda(c,s_1,s_2) \geq 1$, then the best correlated equilibrium is the best \Ne, where all players would rather go.

\begin{corollary}\label{corollary:no.mediation}
For any \extelfarol, if $c > 1$ and $\xlambda(c,s_1,s_2) \geq 1$ then $\mvsym = 1$.
\end{corollary}

\begin{proof}
By Theorem~\ref{thm:optimal}, when $\xlambda(c,s_1,s_2) \geq 1$, $\xstar=1$ and $p=0$. 
Now we prove that if $\xlambda(c,s_1,s_2) \geq 1$, then the best correlated equilibrium is the best \SpacedNe of the case $f(1) < 1$ in Lemma \ref{lem:bestnash}.
To do so, we prove that $\xlambda(c,s_1,s_2) \geq 1 \Rightarrow f(1) < 1$.

Now assume by way of contradiction that 
$
\xlambda(c,s_1,s_2) \geq 1 \Rightarrow f(1) \geq 1.
$
Recall that $f(1) = s_2(1-\frac{c}{s_1})$. Then
$
\xlambda(c,s_1,s_2) \geq 1 \Rightarrow \frac{c}{s_1}+\frac{1}{s_2} \leq 1,$
or equivalently, 
$
\xlambda(c,s_1,s_2) \geq 1 \Rightarrow \frac{c}{s_1}+\frac{1}{s_2} \leq \xlambda(c,s_1,s_2).
$
Also recall that $\xlambda(c,s_1,s_2) = c(\frac{1}{s_1}+\frac{1}{s_2})- \sqrt{\frac{c(\frac{1}{s_1}+\frac{1}{s_2})(c-1)}{s_2}}$.
Thus, we have:

\begin{eqnarray*}
\xlambda(c,s_1,s_2) \geq 1 && \Rightarrow	 \frac{c}{s_1}+\frac{1}{s_2} \leq c(\frac{1}{s_1}+\frac{1}{s_2})- \sqrt{\frac{c(\frac{1}{s_1}+\frac{1}{s_2})(c-1)}{s_2}}\\
	 && \Rightarrow  s_2 \cdot \frac{c}{s_1} \leq -1,
\end{eqnarray*}
which contradicts since $s_1,s_2$ and $c$ are all positive.
Therefore, for $c > 1$ and $\xlambda(c,s_1,s_2) \geq 1$, $MV$ must be equal to $1$. \qed
\end{proof}

%

Now we show that $\mvsym$ and $\evsym$ can be unbounded in the following corollaries.
\begin{corollary}\label{corollary:unboundedmv}
For any $(2+\epsilon,\frac{2+\epsilon}{1-\epsilon}, \frac{1}{\epsilon})$-El Farol game, as $\epsilon \to 0$, $\mvsym \to \infty$.
\end{corollary}

\begin{proof}
For any $(2+\epsilon,\frac{2+\epsilon}{1-\epsilon}, \frac{1}{\epsilon})$-El Farol game, we have $f(1) = 1$. 
By Theorem \ref{thm:optimal}, we obtain $\xstar=1-\epsilon$, $f(\xstar) = 0$ and $p=\frac{\epsilon}{1+2\epsilon}$ for $\epsilon \leq \frac{1}{2}(\sqrt{3}-1)$.
Thus we have 
\[
\lim_{\epsilon \to 0} \mvsym = \lim_{\epsilon \to 0}\frac{\min{(f(1),1)}}{\frac{\epsilon}{1+2\epsilon}+\epsilon(\frac{1+\epsilon}{1+2\epsilon})} = \infty.
\] \qed
\end{proof}

\begin{corollary}\label{corollary:unboundedev}
For any $(1+\epsilon,\frac{1+\epsilon}{1-\epsilon}, \frac{1}{\epsilon})$-El Farol game, as $\epsilon \to 0$, $\evsym \to \infty$.
\end{corollary}

\begin{proof}
For any $(1+\epsilon,\frac{1+\epsilon}{1-\epsilon}, \frac{1}{\epsilon})$-El Farol game, by Theorem \ref{thm:optimal}, we obtain $\xstar=1+\epsilon^2-\epsilon \sqrt{1+\epsilon^2}$ and $f(\xstar) = 1+\epsilon -\sqrt{1-\epsilon^2}$.
Then we have
\[ p=\frac{(1-(1+\epsilon^2-\epsilon \sqrt{1+\epsilon^2}))(1-(1+\epsilon -\sqrt{1-\epsilon^2}))}{(1-(1+\epsilon^2-\epsilon \sqrt{1+\epsilon^2}))(1-(1+\epsilon -\sqrt{1-\epsilon^2}))+\epsilon}.
\]
Also we have $\ystar = 1-\epsilon$ and $f(\ystar) = 0$ for $\epsilon \leq \frac{1}{2}$. 
Thus we have 
\[
\lim_{\epsilon \to 0} \evsym = \lim_{\epsilon \to 0} 
\frac{p+(1-p)(\xstar f(\xstar)+(1-\xstar))}{\ystar f(\ystar)+(1-\ystar)}
= \infty.
\] \qed
\end{proof}

\begin{figure}
\centerline{
\includegraphics[width=0.45\textwidth,natwidth=\fgxnatwidth,natheight=\fgxnatheight]
{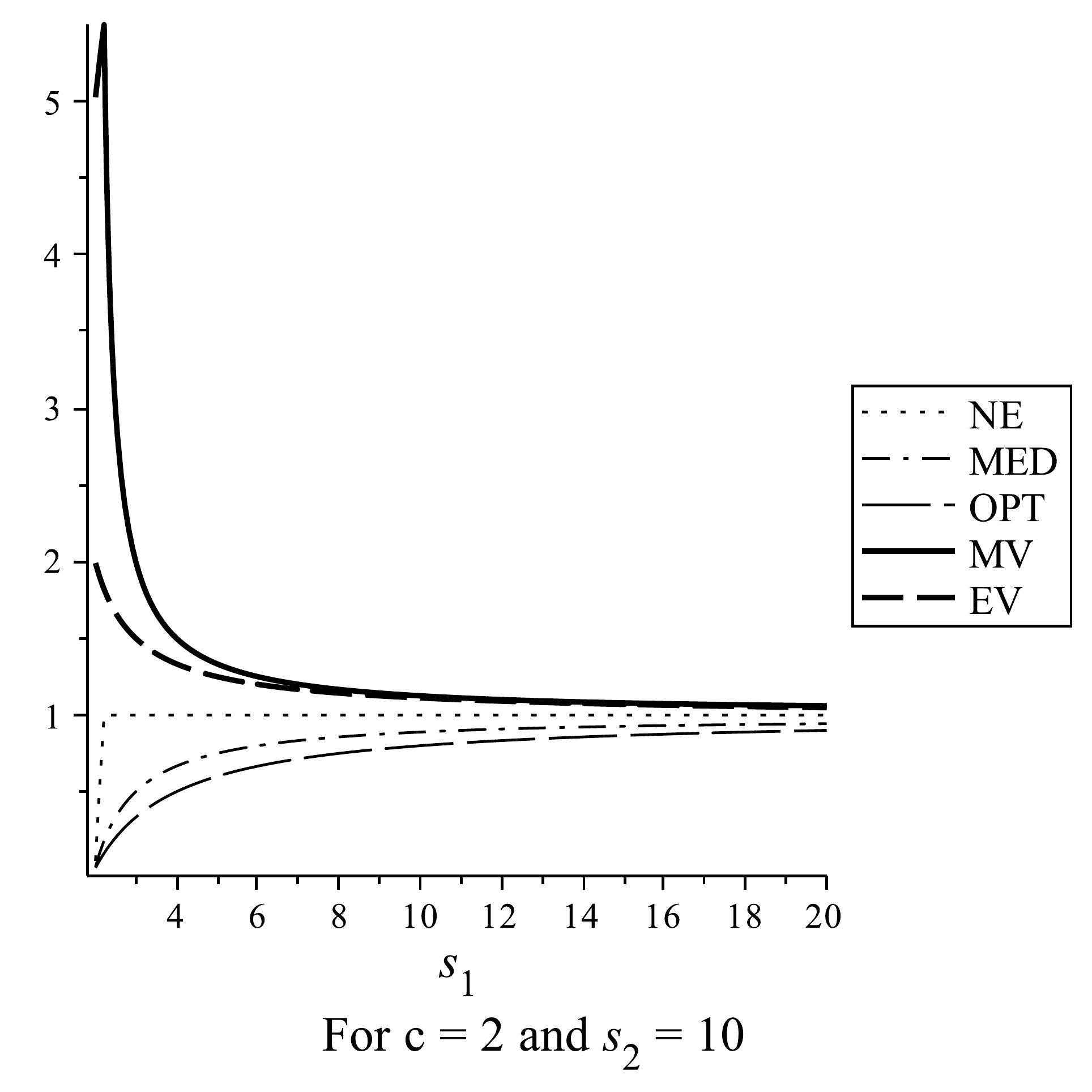}
\includegraphics[width=0.45\textwidth,natwidth=\fgxnatwidth,natheight=\fgxnatheight]
{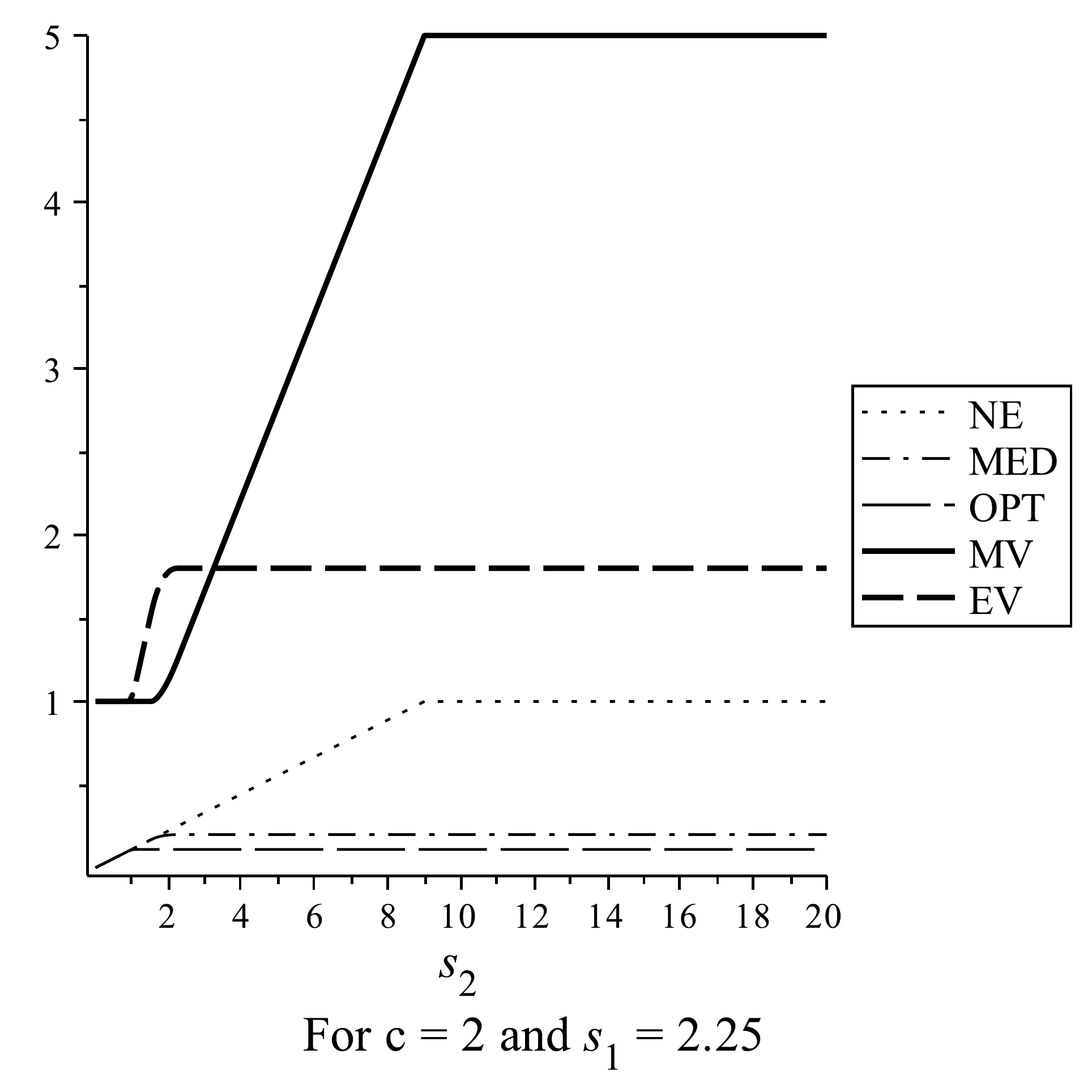}}
\caption{NE, MED, OPT, MV and EV with respect to $s_1$ and $s_2$.}
\label{fig:s1_s2}
\end{figure}

\begin{figure}
\centerline{
\includegraphics[width=0.45\textwidth,natwidth=\plotnatwidth,natheight=\plotnatheight]
{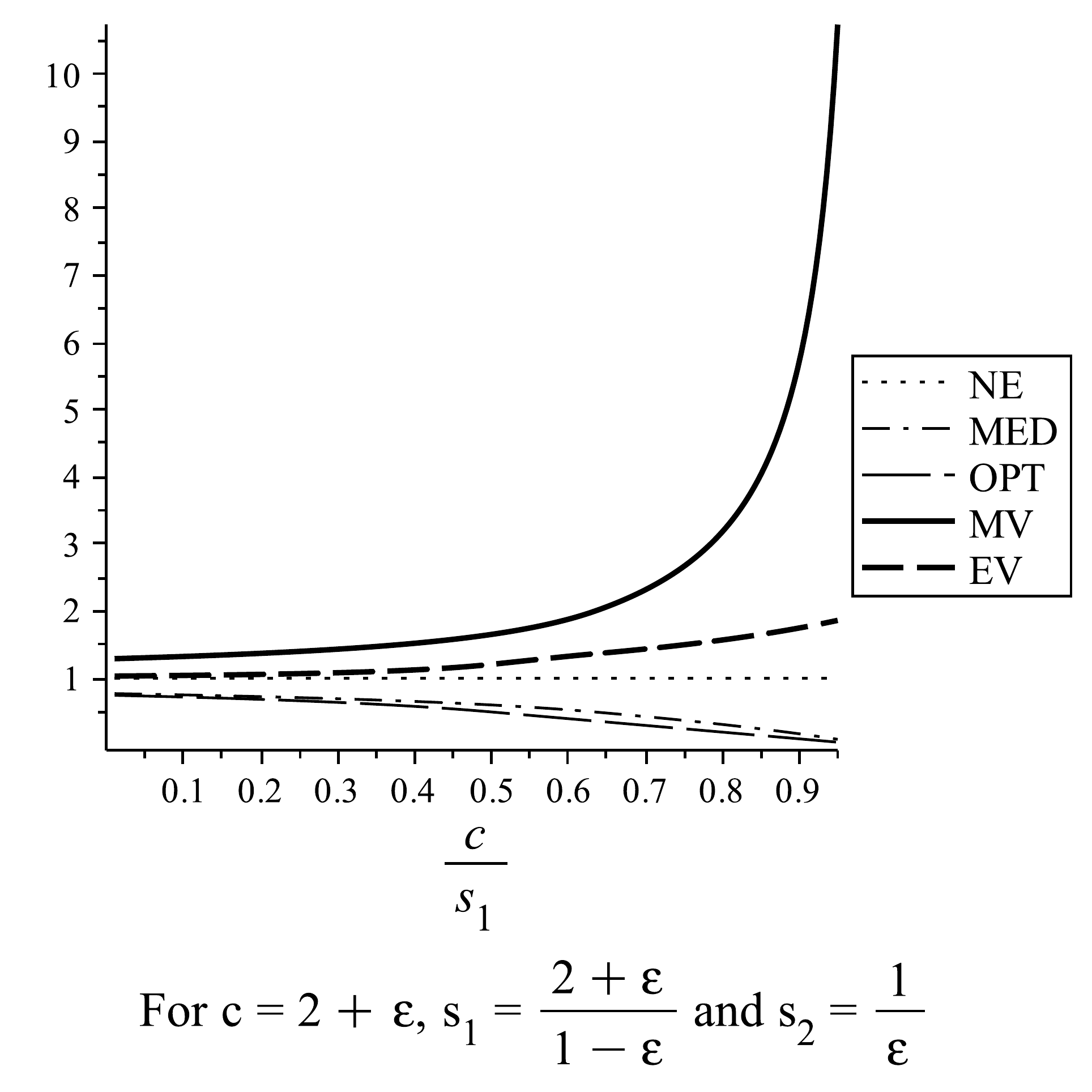}
\includegraphics[width=0.45\textwidth,natwidth=\plotnatwidth,natheight=\plotnatheight]
{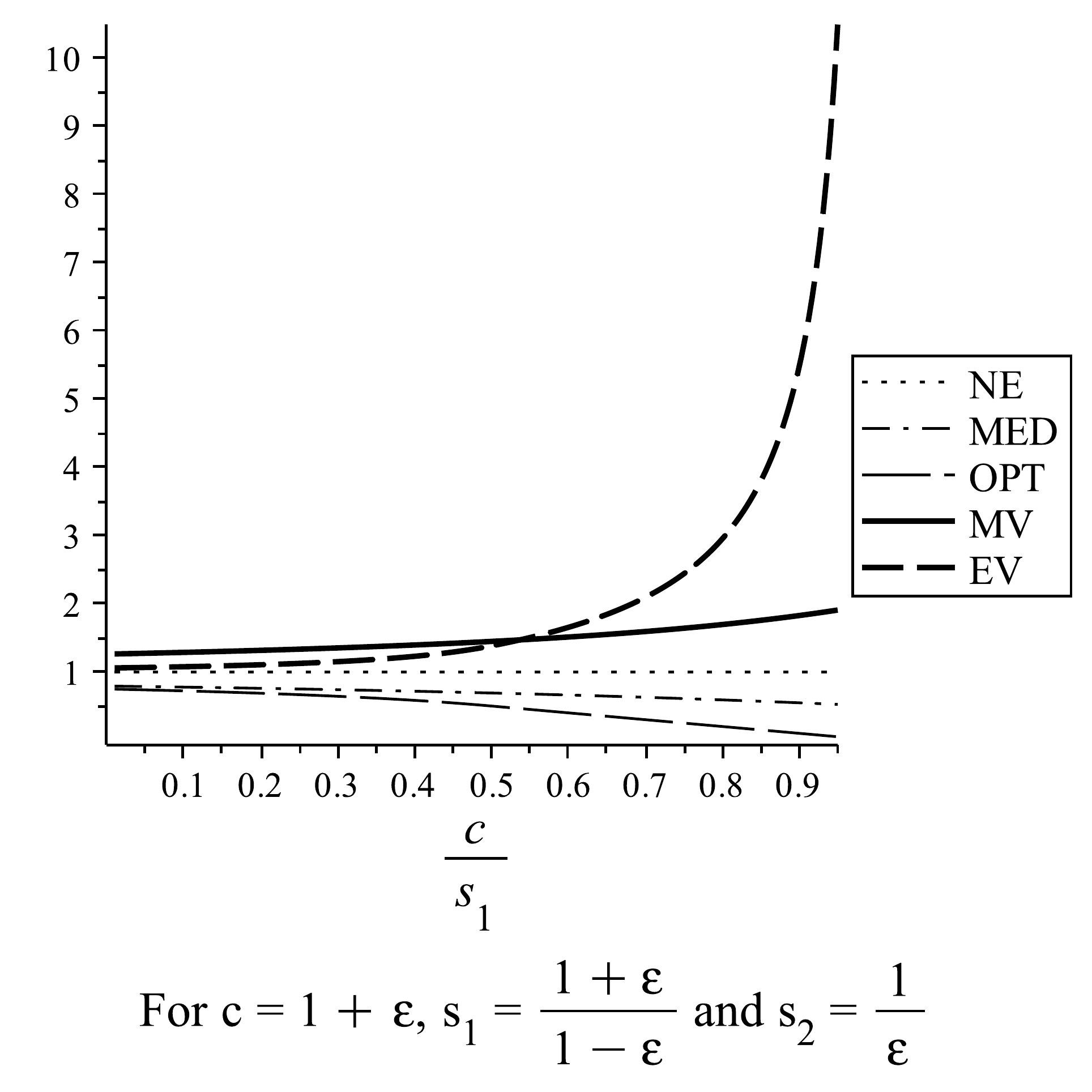}}
\caption{NE, MED, OPT, MV and EV with respect to $c/s_1$.}
\label{fig:c_s1}
\end{figure}

Based on these results, we show in Figures \ref{fig:s1_s2} and \ref{fig:c_s1} the social cost of the best Nash equilibrium (NE), the expected social cost of our optimal mediator (MED) and the optimal social cost (OPT), normalized by $n$, with respect to $s_1$, $s_2$ and $c/s_1$. 
Also we show the corresponding \mv ~(\mvsym) and \ev ~(\evsym).

In Figure \ref{fig:s1_s2}, the left plot shows that for $c=2$ and $s_2=10$, the values of NE, MED, OPT increase, each up to a certain point, when $s_1$ increases; however, the values of \mvsym ~and \evsym ~decrease when $s_1$ increases.
Moreover, \mvsym ~reaches its peak at the point where the best Nash equilibrium starts to remain constant with respect to $s_1$. 
In the right plot, we set $c = 2$ and $s_1 = 2.25$; it shows that the values of NE, MED, OPT, \mvsym ~and \evsym ~increase, each up to a certain point, when $s_2$ increases.

Figure \ref{fig:c_s1} illustrates Corollaries \ref{corollary:unboundedmv} and \ref{corollary:unboundedev}, and it shows how fast \mvsym ~and \evsym ~go to infinity with respect to $c/s_1$, where $c/s_1 = 1 - \epsilon$.
The left plot shows that for any $(2+\epsilon,\frac{2+\epsilon}{1-\epsilon}, \frac{1}{\epsilon})$-El Farol game, as $c/s_1 \to 1$ ($\epsilon \to 0$), $MV \to \infty$ and $EV \to 2$. 
In the right plot, for any $(1+\epsilon,\frac{1+\epsilon}{1-\epsilon}, \frac{1}{\epsilon})$-El Farol game, as $c/s_1 \to 1$ ($\epsilon \to 0$), $EV \to \infty$ and $MV \to 2$.

Note that for any \extelfarol, if $c/s_1 = 1$, then the best correlated equilibrium is at which all players would rather go with a social cost of $0$, that is the best \SpacedNe as well.
Therefore, once $c/s_1$ is equal to $1$, $\mvsym$ drops to $1$.

\subsection*{The cost to stay assumption}\label{section_thecosttostayassumption}

Now we justify our assumption that the cost to stay is unity.
Let $(c',s'_1,s'_2,t')$-El Farol game be a variant of \extelfarol, where $0 < c' < s'_1$, $s' > 0$ and the cost to stay is $t'>0$.
If $x$ is the fraction of players to go, then the cost $f'(x)$ for any player to go is as follows:
$$
f'(x) = \left\{ 
\begin{array}{l l}
  c'- s'_1 x & \quad \mbox{$0 \leq x \leq \frac{c'}{s'_1}$,}\\
  s'_2 (x - \frac{c'}{s'_1}) & \quad \mbox{$\frac{c'}{s'_1} \leq x \leq 1$.}\\ 
\end{array} \right.
$$
The following lemma shows that any $(c',s'_1,s'_2,t')$-El Farol game can be reduced to a \extelfarol .

\begin{lemma}\label{lem:t_normalization}
Any $(c',s'_1,s'_2,t')$-El Farol game can be reduced to a \extelfarolws that has the same \mv ~and \ev, where $c = \frac{c'}{t'}, s_1 = \frac{s'_1}{t'}$ and $s_2 = \frac{s'_2}{t'}$.
\end{lemma}

\begin{proof}
In a manner similar to Theorem~(\ref{thm:optimal}), for any $(c',s'_1,s'_2,t')$-El Farol game, if $c > t'$, then the best correlated equilibrium is $\D\{(\conf{0},p'),(\conf{x'},1-p')\}$, 
where 
$
\xlambda'(c',s'_1,s'_2,t') = c'(\frac{1}{s'_1}+\frac{1}{s'_2})- \sqrt{\frac{c'(\frac{1}{s'_1}+\frac{1}{s'_2})(c'-t')}{s'_2}};
$
$$
x' = \left\{ 
\begin{array}{l l}
  \xlambda'(c',s'_1,s'_2,t') & \quad \mbox{if $\frac{c'}{s'_1}  \leq  \xlambda'(c',s'_1,s'_2,t') < 1$,}\\
  \frac{c'}{s'_1} & \quad \mbox{if $ \xlambda'(c',s'_1,s'_2,t') < \frac{c'}{s'_1}$,}\\ 
  1 & \quad \mbox{otherwise} .\\ 
\end{array} \right.
$$
and 
$
p'=\frac{(1-x')(t'-f(x'))}{(1-x')(t'-f(x'))+c'-t'}.
$
Moreover,
\begin{enumerate}

\item[1)] the Mediation Value ($MV'$) is
$
\frac{\min{(f'(1), t')}}{p't'+(1-p')(x'f(x')+(1-x')t')}
\ and
$
\item[2)] the Enforcement Value ($EV'$) is
$
\frac{p't'+(1-p')(x'f(x')+(1-x')t')}{y'f(y')+(1-y')t'},
$ 
where 
$$
y' = \left\{ 
\begin{array}{l l}
  \frac{1}{2}(\frac{c'}{s'_1}+\frac{t'}{s'_2}) & \quad \mbox{if $\frac{c'}{s'_1}  \leq  \frac{1}{2}(\frac{c'}{s'_1}+\frac{t'}{s'_2}) \leq 1$,}\\
  \frac{c'}{s'_1} & \quad \mbox{if $ \frac{t'}{s'_2} < \frac{c'}{s'_1}$,}\\ 
  1 & \quad \mbox{$ otherwise$.}\\ 
\end{array} \right. .
$$
\end{enumerate}

Similarly, for $c \leq t'$, we have $MV' = 1$ and $EV' = \frac{\min{(f'(1), t')}}{y'f(y')+(1-y')t'}$.

For both cases, by Theorem \ref{thm:optimal}, if we set $c=c'/t'$, $s_1=s'_1/t'$ and $s_2=s'_2/t'$, then we have $f'(1) = f(1) \cdot t'$; also we get $y' = \ystar$ and  $\lambda'(c',s'_1,s'_2,t')  = \lambda(c,s_1,s_2)$.
This implies that $f'(y') = f(\ystar) \cdot t'$ and $x' = \xstar$; which in turn $f'(x') = f(\xstar) \cdot t'$ and $p' = p$.
Thus, we obtain $MV'=MV$ and $EV'=EV$. \qed
\end{proof}

\section{Conclusion}\label{sec:conclusion}
We have extended the traditional \elfarolws to have both negative and positive network effects. We have described an optimal mediator, and we have measured the \mv ~and the \ev ~to completely characterize the benefit of our mediator with respect to the best Nash equilibrium and the optimal social cost.

Several open questions remain including the following: can we generalize our results for our game where the players choose among $k>2$ actions? 
How many configurations are required to design an optimal mediator when there are $k>2$ actions?
Another problem is characterizing the \mvsym ~and \evsym ~values for our game with the more powerful mediators in \cite{Ashlagi:2007,Forgo2010,Forgo2010-2,Monderer:2009,Peleg:2007,RT,RT3,RT2,Tennenholtz:2008}. 
How much would these more powerful mediators reduce the social cost over our type of weaker mediator?

\bibliographystyle{splncs}
\bibliography{elfarol}

\begin{thebibliography}{10}

\bibitem{Koutsoupias1999}
Koutsoupias, E., Papadimitriou, C.:
\newblock Worst-case equilibria.
\newblock In: Proceedings of the 16th Annual Symposium on Theoretical Aspects
  of Computer Science. STACS'99, Berlin, Heidelberg, Springer-Verlag (1999)
  404--413

\bibitem{Ashlagi:2007}
Ashlagi, I., Monderer, D., Tennenholtz, M.:
\newblock Mediators in position auctions.
\newblock In: Proceedings of the 8th ACM Conference on Electronic Commerce. EC
  '07, New York, NY, USA, ACM (2007)  279--287

\bibitem{Diaz2009}
D\'{\i}az, J., Mitsche, D., Rustagi, N., Saia, J.:
\newblock On the power of mediators.
\newblock In: Proceedings of the 5th International Workshop on Internet and
  Network Economics. WINE '09, Berlin, Heidelberg, Springer-Verlag (2009)
  455--462

\bibitem{Forgo2010}
Forg{\'o}, F.:
\newblock A generalization of correlated equilibrium: A new protocol.
\newblock Mathematical Social Sciences \textbf{60}(3) (2010)  186--190

\bibitem{Forgo2010-2}
Forg{\'o}, F.:
\newblock Measuring the power of soft correlated equilibrium in 2-facility
  simple non-increasing linear congestion games.
\newblock Central European Journal of Operations Research (2012)  1--17

\bibitem{Monderer:2009}
Monderer, D., Tennenholtz, M.:
\newblock Strong mediated equilibrium.
\newblock Artif. Intell. \textbf{173}(1) (January 2009)  180--195

\bibitem{Peleg:2007}
Peleg, B., Procaccia, A.D.:
\newblock Implementation by mediated equilibrium.
\newblock International Journal of Game Theory \textbf{39}(1-2) (2010)
  191--207

\bibitem{RT}
Rozenfeld, O., Tennenholtz, M.:
\newblock Strong and correlated strong equilibria in monotone congestion games.
\newblock In: Proceedings of the 2nd International Workshop on Internet and
  Network Economics. WINE'06, Berlin, Heidelberg, Springer-Verlag (2006)
  74--86

\bibitem{RT3}
Rozenfeld, O., Tennenholtz, M.:
\newblock Group dominant strategies.
\newblock In: Proceedings of the 3rd International Workshop on Internet and
  Network Economics. WINE'07, Berlin, Heidelberg, Springer-Verlag (2007)
  457--468

\bibitem{RT2}
Rozenfeld, O., Tennenholtz, M.:
\newblock Routing mediators.
\newblock In: Proceedings of the 20th International Joint Conference on
  Artificial Intelligence. IJCAI'07, San Francisco, CA, USA, Morgan Kaufmann
  Publishers Inc. (2007)  1488--1493

\bibitem{Tennenholtz:2008}
Tennenholtz, M.:
\newblock Game-theoretic recommendations: some progress in an uphill battle.
\newblock In: Proceedings of the 7th International Joint Conference on
  Autonomous Agents and Multiagent Systems. AAMAS '08, Richland, SC (2008)
  10--16

\bibitem{Aumann}
Aumann, R.:
\newblock Subjectivity and correlation in randomized games.
\newblock Mathematical Economics \textbf{1} (1974)  67--96

\bibitem{AshlagiMT08}
Ashlagi, I., Monderer, D., Tennenholtz, M.:
\newblock On the value of correlation.
\newblock Journal of Artificial Intelligence Research (JAIR) \textbf{33} (2008)
   575--613

\bibitem{Arthur}
Arthur, B.:
\newblock Bounded rationality and inductive behavior (the el farol problem).
\newblock American Economic Review \textbf{84} (1994)  406--411

\bibitem{CPG}
De~Cara, M., Pla, O., Guinea, F.:
\newblock Competition, efficiency and collective behavior in the ``el farol''
  bar model.
\newblock The European Physical Journal B - Condensed Matter and Complex
  Systems \textbf{10}(1) (1999)  187--191

\bibitem{CMO}
Challet, D., Marsili, M., Ottino, G.:
\newblock Shedding light on el farol.
\newblock Physica A: Statistical Mechanics and Its Applications \textbf{332}
  (2004)  469--482

\bibitem{LAKUA}
Lus, H., Aydin, C., Keten, S., Unsal, H., Atiligan, A.:
\newblock El farol revisited.
\newblock Physica A: Statistical Mechanics and Its Applications
  \textbf{346}(3-4) (2005)  651--656

\bibitem{David:2010}
Easley, D., Kleinberg, J.:
\newblock Networks, Crowds, and Markets: Reasoning About a Highly Connected
  World.
\newblock Cambridge University Press, New York, NY, USA (2010)

\bibitem{aumann1974values}
Aumann, R., Shapley, L.:
\newblock Values of Non-Atomic Games.
\newblock A Rand Corporation Research Study Series. Princeton University Press
  (1974)

\bibitem{Schmeidler1973}
Schmeidler, D.:
\newblock Equilibrium points of nonatomic games.
\newblock Journal of Statistical Physics \textbf{7}(4) (1973)  295--300

\bibitem{Christodoulou:2005}
Christodoulou, G., Koutsoupias, E.:
\newblock On the price of anarchy and stability of correlated equilibria of
  linear congestion games.
\newblock In: Proceedings of the 13th annual European Symposium on Algorithms.
  ESA'05, Berlin, Heidelberg, Springer-Verlag (2005)  59--70

\bibitem{Brandt:2007}
Brandt, F., Fischer, F., Harrenstein, P., Shoham, Y.:
\newblock A game-theoretic analysis of strictly competitive multiagent
  scenarios.
\newblock In: Proceedings of the 20th International Joint Conference on
  Artifical Intelligence. IJCAI'07, San Francisco, CA, USA, Morgan Kaufmann
  Publishers Inc. (2007)  1199--1206

\bibitem{Brandt:2006}
Brandt, F., Fischer, F., Shoham, Y.:
\newblock On strictly competitive multi-player games.
\newblock In: Proceedings of the 21st national conference on Artificial
  intelligence. Volume~1 of AAAI'06., AAAI Press (2006)  605--612

\bibitem{ACY}
Aspnes, J., Chang, K., Yampolskiy, A.:
\newblock Inoculation strategies for victims of viruses and the sum-of-squares
  partition problem.
\newblock Journal of Computer and System Science \textbf{72}(6) (2006)
  1077--1093

\bibitem{MSW}
Moscibroda, T., Schmid, S., Wattenhofer, R.:
\newblock When selfish meets evil: byzantine players in a virus inoculation
  game.
\newblock In: Proceedings of the twenty-fifth annual ACM symposium on
  Principles of Distributed Computing. PODC '06, New York, NY, USA, ACM (2006)
  35--44

\bibitem{Xin13a}
Jiang, A.X., Procaccia, A.D., Qian, Y., Shah, N., Tambe, M.:
\newblock Defender (mis)coordination in security games.
\newblock In: International Joint Conference on Artificial Intelligence
  (IJCAI). (2013)

\bibitem{Bradonjic:2009}
Bradonjic, M., Ercal-Ozkaya, G., Meyerson, A., Roytman, A.:
\newblock On the price of mediation.
\newblock In: Proceedings of the 10th ACM Conference on Electronic Commerce. EC
  '09, New York, NY, USA, ACM (2009)  315--324

\bibitem{Papadimitriou:2008}
Papadimitriou, C.H., Roughgarden, T.:
\newblock Computing correlated equilibria in multi-player games.
\newblock J. ACM \textbf{55}(3) (August 2008)  1--29

\bibitem{ADGH}
Abraham, I., Dolev, D., Gonen, R., Halpern, J.:
\newblock Distributed computing meets game theory: Robust mechanisms for
  rational secret sharing and multiparty computation.
\newblock In: Proceedings of the 25th Annual ACM Symposium on Principles of
  Distributed Computing. PODC '06, New York, NY, USA, ACM (2006)  53--62

\bibitem{ADH}
Abraham, I., Dolev, D., Halpern, J.Y.:
\newblock Lower bounds on implementing robust and resilient mediators.
\newblock In: Proceedings of the 5th Conference on Theory of Cryptography.
  TCC'08, Berlin, Heidelberg, Springer-Verlag (2008)  302--319

\bibitem{Forges1986}
Forges, F.:
\newblock An approach to communication equilibria.
\newblock Econometrica \textbf{54}(6) (1986)  1375--1385

\bibitem{Myerson1986}
Myerson, R.B.:
\newblock Multistage games with communication.
\newblock Econometrica \textbf{54}(2) (March 1986)  323--58

\end{thebibliography}


\pagebreak
\appendix
\section{Appendix - Proof of Theorem \ref{thm:optimal}} \label{sec:proofoftheorem}
First of all, we call a mediator over $k$ configurations when the configuration distribution this mediator uses has $k$ configurations.

The proof has four main parts. The first part is \textit{The Reduction of Mediators for $c > 1$}, where we prove that if $c > 1$, then for any optimal mediator over $k > 2$ configurations, there is a mediator over two configurations that has the same social cost. 
The second part is \textit{The Reduction of Mediators for $c \leq 1$}, where we prove that if $c \leq 1$, then the best correlated equilibrium is the best Nash equilibrium. 
The third part is \textit{An Optimal Mediator}, where we describe an optimal mediator for any arbitrary constants $c, s_1$ and $s_2$. 
Finally, the fourth part is \textit{The Mediation Metrics}, where we measure the \mv ~and the \ev .

Recall that $x_i$ is the fraction of players that are advised to go in configuration $\conf{x_i}$ which is selected with probability $p_i$ in a configuration distribution, $\D\{(\conf{x_1},p_1),.., (\conf{x_k},p_k)\}$, for $1\leq i \leq k$. 
We define $\Deltaf{x_i} = 1-f(x_i)$, where $f(x_i)$ is defined in Equation \eqref{eq:fgx}.

\subsection{The Reduction of Mediators for $c > 1$}
In this section, we consider the case that $c >1$. 

\begin{fact}\label{fact:deltas}
For any mediator over $k$ configurations, and for $1\leq i\leq k$,
$\Deltaf{x_i}>0$ iff 
$(\frac{c-1}{s_1} < x_i< \frac{1}{s_2}+\frac{c}{s_1} \ and \ f(1) \geq 1)$ or $(\frac{c-1}{s_1} < x_i \leq 1 \ and \ f(1) < 1)$; and 
$\Deltaf{x_i}<0$ iff 
$0\leq x_i <\frac{c-1}{s_1}$ or $(\frac{1}{s_2}+\frac{c}{s_1}<x_i \leq 1 \ and \ f(1) > 1)$.
\end{fact}

\begin{proof}
Recall that $\Deltaf{x_i} = 1-f(x_i)$. Then by Equation (\ref{eq:fgx}), we have

\[
\Deltaf{x_i} = \left\{ 
\begin{array}{l l}
  \Delta_1(x_i) & \quad \mbox{$0 \leq x_i \leq \frac{c}{s_1}$,}\\
  \Delta_2(x_i) & \quad \mbox{$\frac{c}{s_1} \leq x_i \leq 1$.}\\ 
\end{array} \right.
\]
where $\Delta_1(x_i) =1-(c- s_1x_i)$ and $\Delta_2(x_i)=1-s_2 (x_i - \frac{c}{s_1})$. Now we make a case analysis:\\
\textbf{Case 1:} $0\leq x_i\leq \frac{c}{s_1}$:
$\Delta_1(x_i)<0 \Longleftrightarrow 0\leq x_i < \frac{c-1}{s_1}$; and $\Delta_1(x_i)>0 \Longleftrightarrow \frac{c-1}{s_1}<x_i\leq \frac{c}{s_1}$.\\
\textbf{Case 2:} $\frac{c}{s_1} \leq x_i \leq 1$:
$\Delta_2(x_i)>0 \Longleftrightarrow 
(\frac{c}{s_1} \leq x_i< \frac{1}{s_2}+\frac{c}{s_1} \ and \ f(1) \geq 1)$ or $(\frac{c}{s_1} \leq x_i \leq 1 \ and \ f(1) < 1)$; and $\Delta_2(x_i)<0 \Longleftrightarrow (\frac{1}{s_2}+\frac{c}{s_1}<x_i \leq 1 \ and \ f(1) > 1)$. \qed
\end{proof}

\begin{fact}\label{fact:const}
$\D\{(\conf{x_1},p_1),..,(\conf{x_k},p_k)\}$ is a correlated equilibrium iff
\begin{equation}\label{D_kconstraints>0}
\sum_{i=1}^{k} p_ix_i\Deltaf{x_i} \geq 0  
\end{equation}
and
\begin{equation}\label{D_kconstraints<0}
\sum_{i=1}^{k} p_i(1-x_i)\Deltaf{x_i} \leq 0.
\end{equation}
\end{fact}

\begin{proof}
Recall that $\cE^i_G$ is the event that the mediator advises player $i$ to go, $C^i_G$ is the cost for player $i$ to go, $\cE^i_S$ is the event that the mediator advises player $i$ to stay, and $C^i_S$ is the cost for player $i$ to stay. Also we will omit the index $i$ since the players are symmetric.

By definition, $\D\{(\conf{x_1},p_1),..,(\conf{x_k},p_k)\}$ is a correlated equilibrium iff 
\begin{eqnarray*}
\EXP{C_S |\cE_G} \geq \EXP{C_G |\cE_G}, \\
\EXP{C_G |\cE_S} \geq \EXP{C_S |\cE_S}.
\end{eqnarray*}
Note that:
\begin{eqnarray*}\label{eq:stay_go}
\EXP{C_S |\cE_G} = 1,
\end{eqnarray*}
\begin{eqnarray*}\label{eq:go_go}
\EXP{C_G |\cE_G} =  \frac{\sum_{i=1}^{k} p_i f(x_i) x_i}{\sum_{i=1}^{k} p_i x_i},
\end{eqnarray*}
\begin{eqnarray*}\label{eq:go_stay}
\EXP{C_G |\cE_S} =  \frac{\sum_{i=1}^{k} p_i f(x_i) (1-x_i)}{\sum_{i=1}^{k} p_i (1-x_i)}
\end{eqnarray*}
and
\begin{eqnarray*}\label{eq:stay_stay}
E(C_S |\cE_S) =  1.
\end{eqnarray*}

Therefore, $\D\{(\conf{x_1},p_1),..,(\conf{x_k},p_k)\}$ is a correlated equilibrium iff
\begin{eqnarray}\label{eqn:const1}
\frac{\sum_{i=1}^{k} p_i f(x_i) x_i}{\sum_{i=1}^{k} p_i x_i}\leq 1
\end{eqnarray}
and 
\begin{eqnarray}\label{eqn:const2}
\frac{\sum_{i=1}^{k} p_i f(x_i) (1-x_i)}{\sum_{i=1}^{k} p_i (1-x_i)} \geq  1.
\end{eqnarray}
By rearranging Inequalities  \eqref{eqn:const1} and \eqref{eqn:const2}, we have
\begin{eqnarray*}
\sum_{i=1}^{k} p_i x_i (1-f(x_i))\geq 0
\end{eqnarray*}
and 
\begin{eqnarray*}
\sum_{i=1}^{k} p_i (1-x_i) (1-f(x_i)) \leq  0.
\end{eqnarray*} \qed
\end{proof}

\begin{fact}\label{fact:socialcost}
The expected social cost of $\D\{(\conf{x_1},p_1),..,(\conf{x_k},p_k)\}$ is
\[
(1 - \sum_{i=1}^{k} p_i x_i \Deltaf{x_i})n.
\]
\end{fact}
\begin{proof}
Let $\Csc{\conf{x_i}}$ be the cost of configuration $\conf{x_i}$ in $\D\{(\conf{x_1},p_1),..,\\(\conf{x_k},p_k)\}$, for $1\leq i \leq k$.
We know that the expected social cost of $\D\{(\conf{x_1},p_1)\\,..,(\conf{x_k},p_k)\}$ is
\[ 
\sum_{i=1}^{k} p_i \Csc{\conf{x_i}}.
\]
We have $\Csc{\conf{x_i}}=(x_if(x_i)+(1-x_i))n$, and since $\Deltaf{x_i}=1-f(x_i)$, it follows that $\Csc{\conf{x_i}}=(1-x_i\Deltaf{x_i})n$. Therefore, the expected social cost is
\[ 
\sum_{i=1}^{k} p_i (1-x_i\Deltaf{x_i})n,
\]
or equivalently,
\[ 
(\sum_{i=1}^{k} p_i - \sum_{i=1}^{k}p_ix_i\Deltaf{x_i})n.
\]
Finally, we note that $\sum_{i=1}^{k} p_i=1$. \qed
\end{proof}

\begin{fact}\label{fact:signedDelta_i}
For any optimal mediator over $k\geq 2$ configurations, $\Deltaf{x_i}\neq 0$ for all $1\leq i\leq k$,
and $\Deltaf{x_u}>0$ and $\Deltaf{x_v}<0$ for some $1\leq u,v \leq k$.
\end{fact}

\begin{proof}
First we show that for any optimal mediator over $k\geq2$ configurations, $\Deltaf{x_i}$ is non-zero for all $1\leq i\leq k$.
Assume by way of contradiction that there is an optimal mediator $M_k$ that uses $\D\{(\conf{x_1},p_1),..,(\conf{x_k},p_k)\}$, and there is some $1\leq j\leq k$ such that $\Deltaf{x_j} = 0$.
Recall that $0 < p_j < 1$.
Now let $\D\{(\conf{x_1},\frac{p_1}{1-p_j}),..,(\conf{x_{j-1}},\frac{p_{j-1}}{1-p_j}),(\conf{x_{j+1}},\frac{p_{j+1}}{1-p_j}),..,(\conf{x_k},\frac{p_{k}}{1-p_j})\}$ be a configuration distribution over $k-1$ configurations.

Since $M_k$ is a mediator and $\Deltaf{x_j} = 0$, Constraints \eqref{D_kconstraints>0} and \eqref{D_kconstraints<0} of \Fact \ref{fact:const} imply that
\[
\sum_{1\leq i \leq k, i\neq j} p_ix_i\Deltaf{x_i} \geq 0  
\]
and
\[
\sum_{1\leq i \leq k, i\neq j} p_i(1-x_i)\Deltaf{x_i} \leq 0.
\]
Now if we multiply both sides of these two constraints by $\frac{1}{1-p_j}$, we have
\[
\sum_{1\leq i \leq k, i\neq j} \frac{p_i}{1-p_j} x_i\Deltaf{x_i}\geq 0  
\]
and
\[
\sum_{1\leq i \leq k, i\neq j} \frac{p_i}{1-p_j}(1-x_i)\Deltaf{x_i}\leq 0.
\]
By \Fact \ref{fact:const}, $\D\{(\conf{x_1},\frac{p_1}{1-p_j}),..,(\conf{x_{j-1}},\frac{p_{j-1}}{1-p_j}),(\conf{x_{j+1}},\frac{p_{j+1}}{1-p_j}),..,(\conf{x_k},\frac{p_{k}}{1-p_j})\}$ is a correlated equilibrium. Let $M_{k-1}$ be a mediator that uses this correlated equilibrium.
By \Fact \ref{fact:socialcost}, the expected social cost of $M_{k-1}$ is 
\[
(1-\frac{1}{1-p_j}\sum_{1\leq i \leq k, i\neq j} p_ix_i\Deltaf{x_i})n,
\]
and since $\Deltaf{x_j} = 0$, the expected social cost of $M_k$ is 
\[
(1-\sum_{1\leq i \leq k, i\neq j} p_ix_i\Deltaf{x_i})n.
\]
We know that $0<p_j<1$ implies $\frac{1}{1-p_j}>1$.
Therefore, the expected social cost $M_{k-1}$ is less than the expected social cost of $M_k$. This contradicts the fact that $M_k$ is optimal.\\
Recall that $0<p_i<1$ and $0\leq x_i\leq 1$ for all $1\leq i \leq k$.
By \Fact \ref{fact:const}, Constraint \eqref{D_kconstraints>0} implies that there exists $u$ such that $\Deltaf{x_u}>0$ for $1\leq u \leq k$. Similarly, Constraint \eqref{D_kconstraints<0} implies that there exists $v$ such that $\Deltaf{x_v}<0$ for $1\leq v \leq k$. \qed
\end{proof}

\begin{fact}\label{fact:implication}
Any optimal mediator that uses $\D\{(\conf{x_1},p_1),..,(\conf{x_j},p_j),..,\\(\conf{x_k},p_k)\}$, where $k\geq 2$, has
\[
\sum_{1\leq i \leq k, i\neq j} p_i\Deltaf{x_i}(x_i-x_j) \geq 0, 1\leq j\leq k.
\]
\end{fact}

\begin{proof}
Let $M_k$ be an optimal mediator that uses $\D\{(\conf{x_1},p_1),..,(\conf{x_j},p_j),..,\\(\conf{x_k},p_k)\}$. 
We know by \Fact \ref{fact:signedDelta_i} that any configuration, $\conf{x_j}$, has either $\Deltaf{x_j}<0$ or $\Deltaf{x_j}>0$, for $1\leq j \leq k$. 
Now fix any $1\leq j \leq k$, and do a case analysis for $\Deltaf{x_j}$.\\
\textbf{Case 1: } If $\Deltaf{x_j}<0$, then by repeated application of \Fact \ref{fact:const} we have
\begin{equation}\label{tmp:compatible11}
\frac{\sum_{1\leq i \leq k, i\neq j} p_i\Deltaf{x_i}(1-x_i)}{(1-x_j)|\Deltaf{x_j}|} 
\leq 
p_j
\leq 
\frac{\sum_{1\leq i \leq k, i\neq j} p_i\Deltaf{x_i}x_i}{x_j|\Deltaf{x_j}|}
\end{equation}
Removing $p_j$ from Inequality \eqref{tmp:compatible11} and rearranging, we get
\[
x_j|\Deltaf{x_j}|\sum_{1\leq i \leq k, i\neq j} p_i\Deltaf{x_i}(1-x_i) 
\leq 
(1-x_j)|\Deltaf{x_j}|\sum_{1\leq i \leq k, i\neq j} p_i\Deltaf{x_i}x_i.
\]
By canceling the common terms, we have
\[
\sum_{1\leq i \leq k, i\neq j} p_i\Deltaf{x_i}x_j 
\leq 
\sum_{1\leq i \leq k, i\neq j} p_i\Deltaf{x_i}x_i.
\]
\textbf{Case 2: } If $\Deltaf{x_j}>0$, then similarly by repeated application of \Fact \ref{fact:const} we have
\begin{equation}\label{tmp:compatible21}
\frac{-\sum_{1\leq i \leq k, i\neq j} p_i\Deltaf{x_i}x_i}{x_j\Deltaf{x_j}}
 \leq 
p_j
\leq 
\frac{-\sum_{1\leq i \leq k, i\neq j} p_i\Deltaf{x_i}(1-x_i)}{(1-x_j)\Deltaf{x_j}}
\end{equation}
Removing $p_j$ from Inequality \eqref{tmp:compatible21} and rearranging, we get
\[
x_j\Deltaf{x_j}\sum_{1\leq i \leq k, i\neq j} p_i\Deltaf{x_i}(1-x_i) \leq (1-x_j)\Deltaf{x_j}\sum_{1\leq i \leq k, i\neq j} p_i\Deltaf{x_i}x_i.
\]
By canceling the common terms, we have
\[
\sum_{1\leq i \leq k, i\neq j} p_i\Deltaf{x_i}x_j \leq \sum_{1\leq i \leq k, i\neq j} p_i\Deltaf{x_i}x_i.
\]
Since $j$ is any value between $1$ and $k$, this implies the statement of the lemma for every such $j$. \qed
\end{proof}

\begin{fact}\label{fact:delta<0-x=0}
Consider any mediator $M_k$ that uses $\D\{(\conf{x_1},p_1),..,(\conf{x_j},p_j),..,\\(\conf{x_k},p_k)\}$, where $0<x_j< \frac{c-1}{s_1}$. Then there exists a mediator $M'_k$ of less expected social cost, which uses $\D\{(\conf{x_1},p_1),..,(\conf{x'_j},p_j),..,(\conf{x_k},p_k)\}$, where $x'_j=0$.
\end{fact}
\begin{proof}
Let $M_k$ be a mediator that uses $\D\{(\conf{x_1},p_1),..,(\conf{x_j},p_j),..,(\conf{x_k},p_k)\}$, where $0<x_j< \frac{c-1}{s_1}$. 
By \Fact \ref{fact:const}, we have

\begin{equation}\label{eq:fact:delta<0-x=0-const1}
p_jx_j\Deltaf{x_j} + \sum_{1\leq i\leq k, i\neq j} p_ix_i\Deltaf{x_i} \geq 0  
\end{equation}
and

\begin{equation}\label{eq:fact:delta<0-x=0-const2}
p_j(1-x_j)\Deltaf{x_j} + \sum_{1\leq i\leq k, i\neq j} p_i(1-x_i)\Deltaf{x_i} \leq 0  .
\end{equation}

Since $0<x_j< \frac{c-1}{s_1}$, by \Fact \ref{fact:deltas}, $\Deltaf{x_j}<0$.
Now let $\D\{(\conf{x_1},p_1),..,\\(\conf{x'_j},p_j),..,(\conf{x_k},p_k)\}$ be a configuration distribution that has $x'_j=0$. 
Thus, we have $p_jx'_j\Deltaf{x'_j}=0$ and $p_jx_j\Deltaf{x_j}<0$. 
By Inequality \eqref{eq:fact:delta<0-x=0-const1}, we have 
\begin{equation}\label{eq:med1}
p_jx'_j\Deltaf{x'_j} + \sum_{1\leq i\leq k, i\neq j} p_ix_i\Deltaf{x_i} > 0.
\end{equation}
We know that $\Deltaf{x'_j} <\Deltaf{x_j} < 0$ and $(1-x'_j) > (1-x_j) > 0$, so we have $(1-x'_j)\Deltaf{x'_j}<(1-x_j)\Deltaf{x_j}$. 
By Inequality \eqref{eq:fact:delta<0-x=0-const2}, we get
\begin{equation}\label{eq:med2}
p_j\Deltaf{x'_j} + \sum_{1\leq i\leq k, i\neq j} p_i(1-x_i)\Deltaf{x_i} < 0. 
\end{equation}

Now by \Fact \ref{fact:const} and Inequalities \eqref{eq:med1} and \eqref{eq:med2}, $\D\{(\conf{x_1},p_1),..,(\conf{x'_j},p_j)\\,..,(\conf{x_k},p_k)\}$ is a correlated equilibrium. Let $M'_k$ be a mediator that uses this correlated equilibrium. 
By \Fact \ref{fact:socialcost}, and since $x'_j=0$, the expected social cost of $M'_k$ is
 \[
(1 - \sum_{1\leq i\leq k, i\neq j} p_i x_i \Deltaf{x_i})n.
\]
Moreover, by \Fact \ref{fact:socialcost}, the expected social cost of $M_k$ is 
\[
((1 - \sum_{1\leq i\leq k, i\neq j} p_i x_i \Deltaf{x_i}) - p_jx_j\Deltaf{x_j})n.
\]
Since $\Deltaf{x_j}<0$ and $x_j>0$, the expected social cost of $M'_k$ is less than the expected social cost of $M_k$. \qed
\end{proof}

\begin{fact}\label{fact:delta>0-x=c/s_1}
For $f(1) \geq 1$, consider any mediator $M_k$ that uses $\D\{(\conf{x_1},p_1),..,\\(\conf{x_j},p_j) , .. , (\conf{x_k},p_k)\}$, where $\frac{c-1}{s_1}<x_j < \frac{c}{s_1}$. 
Then there exists a mediator $M'_k$ of less expected social cost, which uses $\D\{(\conf{x_1},p_1),..,(\conf{x'_j},p_j),..,(\conf{x_k},p_k)\}$, where $\frac{c}{s_1} < x'_j< \frac{c}{s_1}+\frac{1}{s_2}$ and $f(x'_j)=f(x_j)$.
\end{fact}
\begin{proof}
Let $M_k$ be a mediator that uses $\D\{(\conf{x_1},p_1),..,(\conf{x_j},p_j),..,(\conf{x_k},p_k)\}$, where $\frac{c-1}{s_1}<x_j < \frac{c}{s_1}$. 
By \Fact \ref{fact:const}, we have

\begin{equation}\label{fact:delta>0-x=c/s_1-const1}
p_jx_j\Deltaf{x_j} + \sum_{1\leq i\leq k, i\neq j} p_ix_i\Deltaf{x_i} \geq 0  
\end{equation}
and

\begin{equation}\label{fact:delta>0-x=c/s_1-const2}
p_j(1-x_j)\Deltaf{x_j} + \sum_{1\leq i\leq k, i\neq j} p_i(1-x_i)\Deltaf{x_i} \leq 0  .
\end{equation}

Recall that $\frac{c-1}{s_1}<x_j < \frac{c}{s_1}$ and $f(1)\geq 1$. 
Then $\exists x'_j:$ $\frac{c}{s_1} < x'_j< \frac{c}{s_1}+\frac{1}{s_2}$ and $f(x'_j)=f(x_j)$.
Now let $\D\{(\conf{x_1},p_1),..,(\conf{x'_j},p_j),..,(\conf{x_k},p_k)\}$ be a configuration distribution. 
Since $f(x'_j)=f(x_j)$, $\Delta(x'_j) = \Delta(x_j)$.
We know that $x'_j>x_j$, then $x'_j\Deltaf{x'_j}>x_j\Deltaf{x_i}$.
By Inequality \eqref{fact:delta>0-x=c/s_1-const1}, we obtain 
\[
p_jx'_j\Deltaf{x'_j} + \sum_{1\leq i\leq k, i\neq j} p_ix_i\Deltaf{x_i} > 0. 
\]
Since $(1-x'_j)<(1-x_j)$, we have $(1-x'_j)\Deltaf{x'_j}<(1-x_j)\Deltaf{x_j}$.
By Inequality \eqref{fact:delta>0-x=c/s_1-const2}, we get
\[
p_j(1-x'_j)\Deltaf{x'_j} + \sum_{1\leq i\leq k, i\neq j} p_i(1-x_i)\Deltaf{x_i} < 0  .
\]
Now by \Fact \ref{fact:const}, $\D\{(\conf{x_1},p_1),..,(\conf{x'_j},p_j),..,(\conf{x_k},p_k)\}$ is a correlated equilibrium. 
Let $M'_k$ be a mediator that uses this correlated equilibrium. 
By \Fact \ref{fact:socialcost}, the expected social cost of $M'_k$ is
 \[
((1 - \sum_{1\leq i\leq k, i\neq j} p_i x_i \Deltaf{x_i}) - p_jx'_j\Deltaf{x'_j})n,
\]
and the expected social cost of $M_k$ is 
\[
((1 - \sum_{1\leq i\leq k, i\neq j} p_i x_i \Deltaf{x_i}) - p_jx_j\Deltaf{x_j})n.
\]
Since $p_jx'_j\Deltaf{x'_j}>p_jx_j\Deltaf{x_i}$, the expected social cost of $M'_k$ is less than the expected social cost of $M_k$. \qed
\end{proof}

\begin{fact}\label{fact:f(1)<1,f(x)<=f(1),(c-1)/s_1<x<c/s_1}
For $f(1) < 1$, consider any mediator $M_k$ that uses $\D\{(\conf{x_1},p_1),..,\\(\conf{x_j},p_j) , .. , (\conf{x_k},p_k)\}$, where $\frac{c-1}{s_1} < x_j < \frac{c}{s_1}$ and $f(x_j) \leq f(1)$. 
Then there exists a mediator $M'_k$ of less expected social cost, which uses $\D\{(\conf{x_1},p_1),..,\\(\conf{x'_j},p_j),..,(\conf{x_k},p_k)\}$, where $\frac{c}{s_1} < x'_j \leq 1$ and $f(x'_j)=f(x_j)$. 
\end{fact}
\begin{proof}
We know that for $\frac{c-1}{s_1}<x_j < \frac{c}{s_1}$ and $f(x_j) \leq f(1) < 1$, $\exists x'_j:$ $\frac{c}{s_1} < x'_j \leq 1$ and $f(x'_j)=f(x_j)$.
In a manner similar to the proof of Lemma \ref{fact:delta>0-x=c/s_1}, we prove this Lemma.
\qed
\end{proof}

\begin{fact}\label{fact:f(1)<1,f(x)>f(1),(c-1)/s_1<x<c/s_1}
For $f(1) < 1$, consider any mediator $M_k$ that uses $\D\{(\conf{x_1},p_1),..,\\(\conf{x_j},p_j) , .. , (\conf{x_k},p_k)\}$, where $\frac{c-1}{s_1} < x_j < \frac{c}{s_1}$ and $f(x_j) > f(1)$. 
Then there exists a mediator $M'_k$ of less expected social cost, which uses $\D\{(\conf{x_1},p_1),..,\\(\conf{x'_j},p_j),..,(\conf{x_k},p_k)\}$, where $x'_j = 1$.
\end{fact}
\begin{proof}
Let $M_k$ be a mediator that uses $\D\{(\conf{x_1},p_1),..,(\conf{x_j},p_j),..,(\conf{x_k},p_k)\}$, where $\frac{c-1}{s_1}<x_j < \frac{c}{s_1}$. 
By \Fact \ref{fact:const}, we have

\begin{equation}\label{fact:f(1)<1,f(x)>f(1),(c-1)/s_1<x<c/s_1-const1}
p_jx_j\Deltaf{x_j} + \sum_{1\leq i\leq k, i\neq j} p_ix_i\Deltaf{x_i} \geq 0  
\end{equation}
and

\begin{equation}\label{fact:f(1)<1,f(x)>f(1),(c-1)/s_1<x<c/s_1-const2}
p_j(1-x_j)\Deltaf{x_j} + \sum_{1\leq i\leq k, i\neq j} p_i(1-x_i)\Deltaf{x_i} \leq 0  .
\end{equation}

Now let $\D\{(\conf{x_1},p_1),..,(\conf{x'_j},p_j),..,(\conf{x_k},p_k)\}$ be a configuration distribution, where $x'_j = 1$. 
For $f(x_j) > f(x'_j)$ and $f(x'_j) < 1$, $\Delta(x'_j) > \Delta(x_j)$.
We know that $x'_j>x_j$, then $x'_j\Deltaf{x'_j}>x_j\Deltaf{x_i}$.
By Inequality \eqref{fact:f(1)<1,f(x)>f(1),(c-1)/s_1<x<c/s_1-const1}, we obtain 
\[
p_jx'_j\Deltaf{x'_j} + \sum_{1\leq i\leq k, i\neq j} p_ix_i\Deltaf{x_i} > 0. 
\]
Since $(1-x'_j) = 0$, $(1-x_j) > 0$ and $\Delta(x_j) > 0$, $(1-x'_j)\Deltaf{x'_j}<(1-x_j)\Deltaf{x_j}$.
By Inequality \eqref{fact:f(1)<1,f(x)>f(1),(c-1)/s_1<x<c/s_1-const2}, we get
\[
p_j(1-x'_j)\Deltaf{x'_j} + \sum_{1\leq i\leq k, i\neq j} p_i(1-x_i)\Deltaf{x_i} < 0  .
\]
Now by \Fact \ref{fact:const}, $\D\{(\conf{x_1},p_1),..,(\conf{x'_j},p_j),..,(\conf{x_k},p_k)\}$ is a correlated equilibrium. 
Let $M'_k$ be a mediator that uses this correlated equilibrium. 
By \Fact \ref{fact:socialcost}, the expected social cost of $M'_k$ is
 \[
((1 - \sum_{1\leq i\leq k, i\neq j} p_i x_i \Deltaf{x_i}) - p_jx'_j\Deltaf{x'_j})n,
\]
and the expected social cost of $M_k$ is 
\[
((1 - \sum_{1\leq i\leq k, i\neq j} p_i x_i \Deltaf{x_i}) - p_jx_j\Deltaf{x_j})n.
\]
Since $p_jx'_j\Deltaf{x'_j}>p_jx_j\Deltaf{x_i}$, the expected social cost of $M'_k$ is less than the expected social cost of $M_k$. \qed
\end{proof}

\begin{lemma}\label{lem:onedelta<0-on-slope-s1}
For any \extelfarol , any optimal mediator that uses $\D\{(\conf{x_1},p_1),..,(\conf{x_k},p_k)\}$ has exactly one configuration that has no players advised to go, and any other configuration has at least a $\frac{c}{s_1}$-fraction of players advised to go.
\end{lemma}

\begin{proof}
Let $M_k$ be an optimal mediator over $k\geq 2$ configurations that uses $\D\{(\conf{x_1},p_1),..,(\conf{x_j},p_j),..,(\conf{x_k},p_k)\}$. 
First we prove that $M_k$ must have a configuration where no players advised to go.
We know by \Fact \ref{fact:signedDelta_i} that there exists some $1\leq j\leq k$ where $\Deltaf{x_j} < 0$. By \Fact \ref{fact:deltas}, $\Deltaf{x_j}<0$ iff $(x_j \in (1/s_2 + c/s_1,1] \ and \ f(1) > 1)$ or $x_j \in [0,\frac{c-1}{s_1})$. Now we do a case analysis for $x_j$.\\
\textbf{Case 1:} $x_j \in (1/s_2 + c/s_1,1] \ and \ f(1) > 1$.
Assume by way of contradiction that $M_k$ has no configuration that has less than a $\frac{c-1}{s_1}$-fraction of players advised to go. 
Let $x_q$ be the smallest fraction that is $x_q>1/s_2 + c/s_1$, where $1\leq q\leq k$.
By \Facts \ref{fact:deltas} and \ref{fact:signedDelta_i}, for $1\leq r\leq k$,
\[\Deltaf{x_r} = \left\{ 
\begin{array}{l l}
  >0 & \quad \mbox{if $x_r<x_q$,}\\
  <0 & \quad \mbox{otherwise}.\\ 
\end{array} \right. \]
Note that by the definition of the configuration distribution, if $x_r = x_q$ then $r=q$.
Therefore, we have
\begin{equation}\label{eq:not-a-mediator}
\sum_{1\leq r \leq k, r\neq q} p_r\Deltaf{x_r}(x_r-x_q)  <  0.
\end{equation}
By \Fact \ref{fact:implication}, Inequality \eqref{eq:not-a-mediator} contradicts that $M_k$ is an optimal mediator. Thus, $M_k$ must have a configuration, $\conf{x}$, where $x<\frac{c-1}{s_1}$, and the rest of the argument is as in Case 2.\\
\textbf{Case 2:} $x_j \in [0,\frac{c-1}{s_1})$. 
By \Fact \ref{fact:delta<0-x=0}, and since $M_k$ is an optimal mediator, $x_j=0$.

By the definition of the configuration distribution, $M_k$ has no two configurations that have the same fraction of players that are advised to go. 
So $M_k$ has exactly one configuration, over all the $k$ configurations, that has no players advised to go.

We know that $\Delta(\frac{c-1}{s_1}) = 0$. 
By \Fact \ref{fact:signedDelta_i}, there is no optimal mediator that has a configuration $\conf{\frac{c-1}{s_1}}$.
Now since $M_k$ is an optimal mediator, by \Facts \ref{fact:delta<0-x=0}, \ref{fact:delta>0-x=c/s_1}, \ref{fact:f(1)<1,f(x)<=f(1),(c-1)/s_1<x<c/s_1} and \ref{fact:f(1)<1,f(x)>f(1),(c-1)/s_1<x<c/s_1}, $M_k$ has no configuration in which an $x$-fraction of players is advised to go, where $x \in (0,\frac{c}{s_1})$. \qed
\end{proof}

\begin{fact}\label{fact:x=py+(1-p)z}
For any $\D\{(\conf{x_1},p_1),..,(\conf{x_i},p_i),..,(\conf{x_j},p_j),..,(\conf{x_k},p_k)\}$, and for any arbitrary $x_i$ and $x_j$ such that $x_j>x_i\geq \frac{c}{s_1}$, 
there exists $\D\{(\conf{x_1},p_1),\\..,(\conf{x_{i-1}},p_{i-1}),(\conf{x_{i+1}},p_{i+1}),..,
(\conf{x'_j},p_i+p_j),..,
(\conf{x_k},p_k)\}$, where $x'_j=\frac{p_i}{p_i+p_j}x_i+\frac{p_j}{p_i+p_j}x_j$. Moreover,

\[
1) \ (p_i+p_j)x'_j\Deltaf{x'_j}>p_ix_i\Deltaf{x_i}+p_jx_j\Deltaf{x_j}.
\]
\[
2) \ (p_i+p_j)(1-x'_j)\Deltaf{x'_j} < p_i(1-x_i)\Deltaf{x_i}+p_j(1-x_j)\Deltaf{x_j}.
\]
\end{fact}
\begin{proof}
Let $\D\{(\conf{x_1},p_1),..,(\conf{x_i},p_i),..,(\conf{x_j},p_j),..,(\conf{x_k},p_k)\}$ be a configuration distribution that has $x_j>x_i\geq \frac{c}{s_1}$. 

Also let $\D\{(\conf{x_1},p_1),..,(\conf{x_{i-1}},p_{i-1}),(\conf{x_{i+1}},p_{i+1}),..,(\conf{x'_j},p_i+p_j),..,\\(\conf{x_k},p_k)\}$ be a configuration distribution that has $x'_j=\frac{p_i}{p_i+p_j}x_i+\frac{p_j}{p_i+p_j}x_j$. 
We know that $0<p_i,p_j<1$ and $x_j>x_i$. Thus $x_i<x'_j<x_j$.
Assume by way of contradiction that
\[
(p_i+p_j)x'_j\Deltaf{x'_j}\leq p_ix_i\Deltaf{x_i}+p_jx_j\Deltaf{x_j},
\]
or equivalently,
\[
x'_j\Deltaf{x'_j}\leq \frac{p_i}{p_i+p_j}x_i\Deltaf{x_i}+\frac{p_j}{p_i+p_j}x_j\Deltaf{x_j}.
\]
Let $p=\frac{p_i}{p_i+p_j}$, so $1-p=\frac{p_j}{p_i+p_j}$. Then we have
\[
x'_j\Deltaf{x'_j}\leq px_i\Deltaf{x_i}+(1-p)x_j\Deltaf{x_j}.
\]
Recall that for $\frac{c}{s_1}\leq x\leq 1$, $\Deltaf{x}=1-s_2(x-\frac{c}{s_1})$.
Since $\frac{c}{s_1}\leq x_i, x_j, x'_j\leq 1$, we get
\[
x'_j(1-s_2(x'_j-\frac{c}{s_1})) \leq 
px_i(1-s_2(x_i-\frac{c}{s_1}))+
(1-p)x_j(1-s_2(x_j-\frac{c}{s_1})).
\]
Since $x'_j=px_i+(1-p)x_j$, we have
\[
x'_j(-s_2(x'_j-\frac{c}{s_1})) \leq 
px_i(-s_2(x_i-\frac{c}{s_1}))+
(1-p)x_j(-s_2(x_j-\frac{c}{s_1})).
\]
We know that $s_2>0$, and hence dividing by $-s_2$, we get 
\[
x'_j(x'_j-\frac{c}{s_1}) \geq px_i(x_i-\frac{c}{s_1})+
(1-p)x_j(x_j-\frac{c}{s_1}).
\]
Since $-\frac{c}{s_1}x'_j=-\frac{c}{s_1}(px_i+(1-p)x_j)$, we have
\[
x'^2_j \geq px_i^2+
(1-p)x_j^2.
\]
Substituting $x'_j$ by $px_i+(1-p)x_j$, we get
\[
p^2x_i^2+2p(1-p)x_ix_j+(1-p)^2x_j^2 \geq px_i^2+
(1-p)x_j^2.
\]
By rearranging, we have
\[
p(1-p)(x_j^2-2x_ix_j+x_i^2) \leq 0.
\]
Now since $0<p<1$, we can divide by $p(1-p)$, and we get
\[
(x_j-x_i)^2 \leq 0,
\]
which contradicts since $ x_j\neq x_i$. This proves that 
\begin{equation}\label{eq:cond1}
(p_i+p_j)x'_j\Deltaf{x'_j} > p_ix_i\Deltaf{x_i}+p_jx_j\Deltaf{x_j}.
\end{equation}
Now we prove that $(p_i+p_j)(1-x'_j)\Deltaf{x'_j} < p_i(1-x_i)\Deltaf{x_i}+p_j(1-x_j)\Deltaf{x_j}$. To do so, we first show that $(p_i+p_j)\Deltaf{x'_j} = p_i\Deltaf{x_i}+
p_j\Deltaf{x_j}$. 

We know that 
\begin{eqnarray*}
	&& x'_j=\frac{p_i}{p_i+p_j}x_i+\frac{p_j}{p_i+p_j}x_j \\
	 \Longleftrightarrow && (p_i+p_j)x'_j = p_ix_i+p_jx_j \\
	 \Longleftrightarrow && (p_i+p_j)(x'_j-\frac{c}{s_1}) = p_i(x_i-\frac{c}{s_1})+p_j(x_j-\frac{c}{s_1}) \\
	 \Longleftrightarrow && (p_i+p_j)s_2(x'_j-\frac{c}{s_1}) = p_is_2(x_i-\frac{c}{s_1})+p_js_2(x_j-\frac{c}{s_1}) \\
	 \Longleftrightarrow && (p_i+p_j)(1-s_2(x'_j-\frac{c}{s_1})) = p_i(1-s_2(x_i-\frac{c}{s_1}))+p_j(1-s_2(x_j-\frac{c}{s_1}))
\end{eqnarray*}
Recall that $\Deltaf{x}=1-s_2(x-\frac{c}{s_1})$ when $\frac{c}{s_1}\leq x\leq 1$.
Since $\frac{c}{s_1}\leq x_i, x_j, x'_j\leq 1$, we get
\begin{equation}\label{eq:cond2}
(p_i+p_j)\Deltaf{x'_j} = p_i\Deltaf{x_i}+p_j\Deltaf{x_j}.
\end{equation}
By subtracting \eqref{eq:cond1} from \eqref{eq:cond2}, we obtain
\[
(p_i+p_j)(1-x'_j)\Deltaf{x'_j} < p_i(1-x_i)\Deltaf{x_i}+p_j(1-x_j)\Deltaf{x_j}.
\] \qed
\end{proof}

\begin{lemma}\label{lem:oneDelta_i>0-on-s2-slope}
For any \extelfarol , any optimal mediator that uses $\D\{(\conf{x_1},p_1),..,(\conf{x_k},p_k)\}$ has at most one configuration that has at least a $\frac{c}{s_1}$-fraction of players advised to go, and any other configuration has less than a $\frac{c}{s_1}$-fraction of players advised to go.
\end{lemma}

\begin{proof}
Assume by way of contradiction that there is an optimal mediator 
$M_k$ that uses $\D\{(\conf{x_1},p_1),..,(\conf{x_i},p_i),..,(\conf{x_j},p_j),..,(\conf{x_k},p_k)\}$, where $x_j>x_i\geq \frac{c}{s_1}$.
Let $\D\{(\conf{x_1},p_1),..,(\conf{x_{i-1}},p_{i-1}),(\conf{x_{i+1}},p_{i+1}),..,(\conf{x'_j},p_i+p_j),\\..,(\conf{x_k},p_k)\}$ be a configuration distribution that has $x'_j=\frac{p_i}{p_i+p_j}x_i+\frac{p_j}{p_i+p_j}x_j$.
Since $M_{k}$ is a mediator, by \Fact \ref{fact:const}, we have
\begin{equation}\label{eq:lem:oneDelta_i>0-on-s2-slope-const1}
p_ix_i\Deltaf{x_i}+ p_jx_j\Deltaf{x_j} + \sum_{1\leq r\leq k, r\neq i, r\neq j} p_rx_r\Deltaf{x_r} \geq 0  
\end{equation}
and
\begin{equation}\label{eq:lem:oneDelta_i>0-on-s2-slope-const2}
p_i(1-x_i)\Deltaf{x_i}+ p_j(1-x_j)\Deltaf{x_j} + \sum_{1\leq r\leq k, r\neq i, r\neq j} p_r(1-x_r)\Deltaf{x_r}\leq 0.
\end{equation} 
By \Fact \ref{fact:x=py+(1-p)z}, we have
\begin{equation}\label{eq:lem:oneDelta_i>0-on-s2-slope-cond1}
(p_i+p_j)x'_j\Deltaf{x'_j} > p_ix_i\Deltaf{x_i}+p_jx_j\Deltaf{x_j}
\end{equation}
and 
\begin{equation}\label{eq:lem:oneDelta_i>0-on-s2-slope-cond2}
(p_i+p_j)(1-x'_j)\Deltaf{x'_j} < p_i(1-x_i)\Deltaf{x_i}+p_j(1-x_j)\Deltaf{x_j}.
\end{equation}
By Inequalities \eqref{eq:lem:oneDelta_i>0-on-s2-slope-const1} and \eqref{eq:lem:oneDelta_i>0-on-s2-slope-cond1}, we get
\begin{equation}\label{eq:lem4.11:cond1}
(p_i+p_j)x'_j\Deltaf{x'_j} + \sum_{1\leq r\leq k, r\neq i, r\neq j} p_rx_r\Deltaf{x_r}> 0.  
\end{equation}
Similarly, by Inequalities \eqref{eq:lem:oneDelta_i>0-on-s2-slope-const2} and \eqref{eq:lem:oneDelta_i>0-on-s2-slope-cond2}, we obtain
\begin{equation}\label{eq:lem4.11:cond2}
(p_i+p_j)(1-x'_j)\Deltaf{x'_j} + \sum_{1\leq r\leq k, r\neq i, r\neq j} p_r(1-x_r)\Deltaf{x_r}< 0.
\end{equation}
By \Fact \ref{fact:const} and Inequalities \eqref{eq:lem4.11:cond1} and \eqref{eq:lem4.11:cond2}, $\D\{(\conf{x_1},p_1),..,(\conf{x_{i-1}},p_{i-1}),\\(\conf{x_{i+1}},p_{i+1}),..,(\conf{x'_j},p_i+p_j),..,(\conf{x_k},p_k)\}$ is a correlated equilibrium. Let $M_{k-1}$ be a mediator that uses this correlated equilibrium.
By \Fact \ref{fact:socialcost}, the expected social cost of $M_{k}$ is 
\[
((1 - \sum_{1\leq r\leq k, r\neq i, r\neq j} p_r x_r \Deltaf{x_r}) - p_ix_i\Deltaf{x_i}-p_jx_j\Deltaf{x_j})n,
\]
and the expected social cost of $M_{k-1}$ is
 \[
((1 - \sum_{1\leq r\leq k, r\neq i, r\neq j} p_r x_r \Deltaf{x_r}) -(p_i+p_j)x'_j\Deltaf{x'_j})n.
\]
Since $(p_i+p_j)x'_j\Deltaf{x'_j}>p_ix_i\Deltaf{x_i}+p_jx_j\Deltaf{x_j}$, the expected social cost of $M_{k-1}$ is less than the expected social cost of $M_{k}$. This contradicts that $M_{k}$ is an optimal mediator. \qed
\end{proof}

\begin{lemma}\label{lem:two-configurations}
For any \extelfarol , there exists an optimal mediator that uses $\D\{(\conf{0},p),(\conf{x},1-p)\}$, where $0<p<1$; and $\frac{c}{s_1} \leq x < \frac{1}{s_2}+\frac{c}{s_1}$ if $f(1) \geq 1$, otherwise $\frac{c}{s_1} \leq x \leq 1$.
\end{lemma}
\begin{proof}
By the definition of the configuration distribution, a mediator has at least two configurations. By Lemmas \ref{lem:onedelta<0-on-slope-s1} and \ref{lem:oneDelta_i>0-on-s2-slope}, there exists an optimal mediator that has exactly two configurations. 
The first configuration has no players advised to go, and the second configuration has an $x$-fraction of players advised to go, where $x\geq \frac{c}{s_1}$.
Since the first configuration has $zero$ players advised to go, by \Fact \ref{fact:deltas}, $\Deltaf{0}<0$. 
By \Fact \ref{fact:signedDelta_i}, we must have $\Deltaf{x}>0$. 
We know that $x\geq \frac{c}{s_1}$. 
By \Fact \ref{fact:deltas}, if $f(1) \geq 1$, then $\frac{c}{s_1} \leq x< \frac{1}{s_2}+\frac{c}{s_1}$; otherwise, $\frac{c}{s_1} \leq x \leq 1$.
\qed
\end{proof}

\subsection{The Reduction of Mediators for $c \leq 1$}
Now we consider the case that $c \leq 1$ in the following lemma.
\begin{lemma}\label{lem:c<1,mv=1}
For any \extelfarolws, if $c \leq 1$, then $\mvsym = 1$.
\end{lemma}
\begin{proof}
In a manner similar to Lemma \ref{fact:signedDelta_i}, any optimal mediator over $k \geq 2$ does not have a configuration $\conf{x}$ with $\Delta(x) = 0$.

Also in a manner similar to Lemmas \ref{fact:delta>0-x=c/s_1}, \ref{fact:f(1)<1,f(x)<=f(1),(c-1)/s_1<x<c/s_1} and \ref{fact:f(1)<1,f(x)>f(1),(c-1)/s_1<x<c/s_1}, for any mediator $M_k$ that uses $\D\{(\conf{x_1},p_1),..,(\conf{x_j},p_j),..,(\conf{x_k},p_k)\}$, where $0 \leq x_j < \frac{c}{s_1}$, there exists a mediator $M'_k$ of less expected social cost, which uses $\D\{(\conf{x_1},p_1),..,\\(\conf{x'_j},p_j),..,(\conf{x_k},p_k)\}$, where $\frac{c}{s_1} \leq x'_j< \frac{1}{s_2}+\frac{c}{s_1}$ if $ f(1) \geq 1$; otherwise, $\frac{c}{s_1} \leq x'_j \leq 1$.

Finally, in a manner similar to Lemma \ref{lem:oneDelta_i>0-on-s2-slope}, any optimal mediator has at most one configuration, $\conf{x}$, where $x \geq \frac{c}{s_1}$.

Therefore, for $c \leq 1$, the best correlated equilibrium is a configuration distribution over just one configuration, which is trivially the best Nash equilibrium.
\end{proof}

\subsection{An Optimal Mediator}
We have proved that for any \extelfarol, if $c \leq 1$ then the best correlated equilibrium is the best Nash equilibrium; otherwise, there exists an optimal mediator that is over two configurations. 
Now we describe this mediator in detail.
\begin{lemma}\label{lem:socialcost(x_2)}  
For any \extelfarol, if $c > 1$, then $\D\{(\conf{0},p),(\conf{x},1-p)\}$ is the best correlated equilibrium, where $\xlambda(c,s_1,s_2) = c(\frac{1}{s_1}+\frac{1}{s_2})- \sqrt{\frac{c(\frac{1}{s_1}+\frac{1}{s_2})(c-1)}{s_2}}$,
$$
x = \left\{ 
\begin{array}{l l}
  \xlambda(c,s_1,s_2) & \quad \mbox{if $\frac{c}{s_1}  \leq  \xlambda(c,s_1,s_2) < 1$,}\\
  \frac{c}{s_1} & \quad \mbox{if $ \xlambda(c,s_1,s_2) < \frac{c}{s_1}$,}\\ 
  1 & \quad \mbox{$ otherwise$.}\\ 
\end{array} \right.
$$
and $p=\frac{(1-x)(1-f(x))}{(1-x)(1-f(x))+c-1}$. Moreover, the expected social cost is
$$
(p+(1-p)(x f(x)+(1-x)))n.
$$
\end{lemma}

\begin{proof}
By Lemma \ref{lem:two-configurations}, there exists an optimal mediator $M_2$ that uses $\D\{(\conf{0},p),\\(\conf{x},1-p)\}$, where $\frac{c}{s_1} \leq x < \frac{1}{s_2}+\frac{c}{s_1}$ if $f(1) \geq 1$; otherwise, $\frac{c}{s_1} \leq x \leq 1$.

Now we determine $p$ and $x$ so that $M_2$ is an optimal mediator.

First, we determine $p$.
By Constraint \eqref{D_kconstraints<0} of \Fact \ref{fact:const}, we have
\begin{equation}\label{lem:socialcost(x)-temp1}
p\Deltaf{0} + (1-p)(1-x)\Deltaf{x} \leq 0.
\end{equation}
We know that $c>1$, $\Deltaf{0}=1-c$ and $\Deltaf{x}>0$. By rearranging Inequality \eqref{lem:socialcost(x)-temp1}, we obtain 
\begin{equation}\label{lem:socialcost(x)-p}
p \geq \frac{(1-x)\Deltaf{x}}{(c-1) + (1-x)\Deltaf{x}}.
\end{equation}
Recall that the cost of any configuration, $\conf{x_i}$, is $\Csc{\conf{x_i}}=(1-x_i\Deltaf{x_i})n$. 
Thus $\Csc{\conf{0}}=n$, and $\Csc{\conf{x}}=(1-x\Deltaf{x})n$. 
Since $\Deltaf{x}>0$, $\Csc{\conf{x}}<n$. 
Thus, $\Csc{\conf{x}}<\Csc{\conf{0}}$. 
We know that the social cost of $M_2$ is
\begin{equation}\label{lem:socialcost(x)-socialcost}
p\Csc{\conf{0}}+(1-p)\Csc{\conf{x}}.
\end{equation}
Since $\Csc{\conf{x}}<\Csc{\conf{0}}$, the minimum expected social cost is when $p$ is the smallest possible value in Inequality \eqref{lem:socialcost(x)-p} which is $\frac{(1-x)\Deltaf{x}}{(c-1) + (1-x)\Deltaf{x}}$.

Now we determine $x$. 
By \Fact \ref{fact:socialcost}, the expected social cost of $M_2$ is
\[
(1 -  (1-p) x \Deltaf{x})n.
\]
Since $p=\frac{(1-x)\Deltaf{x}}{(c-1) + (1-x)\Deltaf{x}}$, the expected social cost is then
\[
(1 -  \frac{(c-1)x \Deltaf{x}}{(c-1) + (1-x)\Deltaf{x}})n.
\]

As $M_2$ is an optimal mediator, we minimize its expected social cost with respect to $x$. 
Thus $g(x)$ is maximized with respect to $x$, where 
\[
g(x) = \frac{(c-1)x \Deltaf{x}}{(c-1) + (1-x)\Deltaf{x}}.
\]
Hence, we have
\[
\frac{dg(x)}{dx} = \frac{(c-1)[(c-1 + (1-x)\Deltaf{x})(\Deltaf{x}-s_2x) + x \Deltaf{x}((1-x)s_2+\Deltaf{x})]}{((c-1) + (1-x)\Deltaf{x})^2}.
\]
By rearranging and canceling common terms, we obtain
\[
\frac{dg(x)}{dx} = \frac{(c-1)[(\Deltaf{x})^2+(c-1)\Deltaf{x}-(c-1)s_2x]}{((c-1) + (1-x)\Deltaf{x})^2}.
\]

We know that $\Deltaf{x}>0$, $\frac{c}{s_1}\leq x < \frac{c}{s_1}+\frac{1}{s_2}$, $x \leq 1$ and $c>1$, so the denominator is always positive. 
By setting the numerator to zero and dividing by $c-1$, we get
\begin{equation}\label{maxsw:df1}
(\Deltaf{x})^2+(c-1)\Deltaf{x}-(c-1)s_2x=0
\end{equation}

By solving Equation \eqref{maxsw:df1}, we have $x = c(\frac{1}{s_1}+\frac{1}{s_2})\pm \sqrt{\frac{c(\frac{1}{s_1}+\frac{1}{s_2})(c-1)}{s_2}}$. 
Now let $\xlambda(c,s_1,s_2) = c(\frac{1}{s_1}+\frac{1}{s_2})- \sqrt{\frac{c(\frac{1}{s_1}+\frac{1}{s_2})(c-1)}{s_2}}$ and 
$\bar{\xlambda}(c,s_1,s_2) = (c(\frac{1}{s_1}+\frac{1}{s_2})+ \sqrt{\frac{c(\frac{1}{s_1}+\frac{1}{s_2})(c-1)}{s_2}})$.

Since $ \bar{\xlambda}(c,s_1,s_2) > (\frac{1}{s_2}+\frac{c}{s_1})$, by Lemma \ref{lem:two-configurations}, it is out of range.
Therefore, we have exactly one root $x = \xlambda(c,s_1,s_2)$.


We know $\frac{dg(x)}{dx}\mid_{(x = \frac{c}{s_1})}<0$ iff $\xlambda(c,s_1,s_2)<\frac{c}{s_1}$, and $\frac{dg(x)}{dx}\mid_{(x = 1)}>0$ iff $\xlambda(c,s_1,s_2) > 1$.
Also we know that $\xlambda(c,s_1,s_2) < \frac{c}{s_1}+\frac{1}{s_2}$; and $\frac{c}{s_1} \leq x < \frac{1}{s_2}+\frac{c}{s_1}$ if $f(1) \geq 1$, otherwise, $\frac{c}{s_1} \leq x \leq 1$. 
Therefore, for $\frac{c}{s_1} \leq \xlambda(c,s_1,s_2) \leq 1$, the maximum of $g(x)$ is at $x=\xlambda(c,s_1,s_2)$.
Moreover, if $\xlambda(c,s_1,s_2) <\frac{c}{s_1}$, then the maximum of $g(x)$ is at $x = \frac{c}{s_1}$; and for $\xlambda(c,s_1,s_2) > 1$, the maximum is at $x = 1$.

Recall that the expected social cost of $\D\{(\conf{0},p),(\conf{x},1-p)\}$ is
\[
p\Csc{\conf{0}}+(1-p)\Csc{\conf{x}},
\]
or equivalently,
\[
(p+(1-p)(xf(x)+(1-x)))n.
\] \qed
\end{proof}

\subsection{The Mediation Metrics}
Now we compute the \mv ~and the \ev .
To obtain the \mv ~and \ev; recall that the \mv ~(\mvsym) is the ratio of the minimum social cost over all Nash equilibria to the minimum social cost over all mediators, and the \ev ~is the ratio of the minimum social cost over all mediators to the optimal social cost.

For $c \leq 1$, by Lemma \ref{lem:c<1,mv=1}, $\mvsym = 1$; and by Lemmas \ref{lem:socialoptimum} and \ref{lem:bestnash}, $\evsym = \frac{\min(f(1),1)}{yf(y)+(1-y)}$.

For $c > 1$, by Lemmas \ref{lem:bestnash} and \ref{lem:socialcost(x_2)}, the \mv ~is:
$$
\frac{\min(f(1),1)}{p+(1-p)(xf(x)+(1-x))};
$$
and by Lemmas \ref{lem:socialoptimum} and \ref{lem:socialcost(x_2)}, the \ev ~is:
$$
\frac{p+(1-p)(xf(x)+(1-x))}{yf(y)+(1-y)}.
$$

\end{document}